\documentclass[11pt]{article}
\usepackage{mathpazo}
\usepackage[T1]{fontenc}
\usepackage[margin=1in]{geometry}
\usepackage{amsmath,amssymb,amsfonts,amsthm}
\usepackage{mathtools}
\usepackage{bm}
\usepackage{bbm}
\usepackage[dvipsnames]{xcolor}
\usepackage{xspace}
\usepackage{microtype}
\usepackage{csquotes}
\usepackage{graphicx}
\usepackage{subcaption}
\usepackage{adjustbox}
\usepackage{float}
\usepackage{enumitem}
\setlist{nosep,topsep=0pt,leftmargin=*}
\usepackage{url}
\usepackage{natbib}
\usepackage{alphabeta}

\usepackage{todonotes}
\usepackage{authblk}

\definecolor{myRed}{rgb}{0.82,0.13,0.56}
\definecolor{myBlue}{RGB}{13,55,174}
\usepackage{hyperref}
\hypersetup{
  colorlinks=true,
  citecolor=PineGreen,
  linkcolor=myBlue,
  urlcolor=PineGreen
}
\usepackage[noabbrev,capitalise]{cleveref}

\usepackage{thmtools}
\usepackage{thm-restate}

\usepackage{eqparbox}

\newtheorem{remark}{Remark}
\newtheorem{proofsketch}{Proof Sketch}

\newcommand{\RED}{\textsc{ReD}\xspace}
\def\LM{\textsc{LargeMarket}}
\def\NLM{\textsc{NM-LargeMarket}}
\def\Active{\text{active}}
\def\Brem{B_{\text{rem}}}

\newcommand{\vmax}{v_{\max}}
\def\ALG{\textsc{Alg}}
\def\alg{\textsc{Alg}}
\def\OPT{\textsc{Opt}}

\newcommand{\lp}{\left} 
\newcommand{\rp}{\right}
\newcommand{\Prob}[1]{\mathbb{P}\!\lp[#1\rp]} 
\newcommand{\E}[2]{\mathbb{E}_{#1\!}\lp[#2\rp]}

\usepackage{todonotes}

\theoremstyle{plain}
\newtheorem{theorem}{Theorem}[section]

\newtheorem{lemma}[theorem]{Lemma}

\newtheorem{corollary}[theorem]{Corollary}
\theoremstyle{definition}
\newtheorem{definition}[theorem]{Definition}

\usepackage[ruled]{algorithm2e} 

\SetAlFnt{\small}
\SetAlCapFnt{\small}
\SetAlCapNameFnt{\small}
\SetAlCapHSkip{0pt}
\IncMargin{-\parindent}
\DontPrintSemicolon
\SetKwInput{KwInput}{Input}
\SetKwInput{KwOutput}{Output}
\SetKw{KwRet}{return}
\SetKw{KwWith}{with probability}

\author[1]{Andreas Charalampopoulos}
\author[2,3]{Dimitris Fotakis}
\author[2,3]{Thanos Tolias}

\affil[1]{Columbia University, New York, USA}
\affil[2]{National Technical University of Athens, Greece}
\affil[3]{Archimedes RU, Athena RC, Greece}

\affil[ ]{\texttt{{
\href{mailto:ac5951@columbia.edu}{ac5951@columbia.edu} \quad
\href{mailto:fotakis@cs.ntua.gr}{fotakis@cs.ntua.gr} \quad
\href{mailto:thanostolias@mail.ntua.gr}{thanostolias@mail.ntua.gr}}}}

\title{Repeated Descent:\\ A Framework for Online Budget-Feasible Auctions}
\date{}

\begin{document}

\maketitle

\renewcommand\thefootnote{}\footnotetext{This work has been partially supported by project MIS 5154714 of the National Recovery and
Resilience Plan Greece 2.0 funded by the European Union under the NextGenerationEU Program.}

\begin{abstract}
We study budget feasible procurement auctions, in which $n$ agents, each with a privately held service cost, offer their services to an employer. The employer seeks to maximize a public submodular valuation function over the set of hired agents, while facing a hard budget constraint. We consider an online posted-price setting, in which agents arrive in a uniformly random order (a.k.a. \emph{secretary arrivals}) and the employer must make irrevocable take-it-or-leave-it offers upon their arrival. The employer does not get any feedback about the agent service costs other than whether they accept the offer or not. 

We introduce Repeated Descent (a.k.a. \RED), a deterministic  framework
based on adaptive linear posted pricing. \RED enforces budget feasibility by adaptively adjusting its pricing and balancing each pricing level with the number of agents considered in it. 
Using \RED as the main building block, we obtain a $1046$-competitive posted-price mechanism for online budget feasible auctions with secretary agent arrivals and submodular valuations, thus improving on the previously best known ratio of (Charalampopoulos et al., EC 2025) by several orders of magnitude.
Combining \RED with random subsampling, we obtain the first constant-competitive posted-price budget feasible mechanism for non-monotone submodular valuations. 
On the negative side, we show that every online budget feasible mechanism with XOS valuations has a competitive ratio of $\Omega\!\left(\tfrac{\log n}{(\log\log n)^2}\right)$.
\end{abstract}

\section{Introduction}
\label{s:intro}

We consider the problem of budget-feasible procurement auctions, introduced by~\citet{Singer10}, where a set \(N\) of agents offer their services to an employer. Each agent \(a\) has a private cost \(c(a)\) for offering their service. The employer has a publicly known valuation function \(v:2^N\to \mathbb{R}_{\ge 0}\) and a hard budget \(B\), and aims to procure a subset of agents so as to maximize \(v(S)\) subject to total payments not exceeding \(B\).

In this work we study \emph{online} budget-feasible procurement in the random-order (secretary) model. The instance (costs and valuation) is fixed adversarially, and agents arrive sequentially in a uniformly random permutation. Upon each arrival, the mechanism must make irrevocable decisions without knowledge of future agents.

At its core, the underlying algorithmic problem is an online maximization problem under knapsack (budget) constraints. Given agents' costs \(c(a)\) for each \(a\in N\), the employer aims to maximize \(v(T)\) subject to \(\sum_{a\in T} c(a)\le B\). This is a classical NP-hard optimization problem, and both its offline and online variants have been extensively studied for many different classes of valuation functions. 

In procurement auctions, however, agents may misreport their private costs in an attempt to extract higher payments from the employer. Consequently, we seek to design mechanisms that are \emph{individually rational} and \emph{truthful}. Individual rationality requires that any agent \(a\) selected by the mechanism receives a payment \(p_a\) that covers their cost, i.e., \(p_a \ge c(a)\). Truthfulness requires that reporting the true cost is a dominant strategy for each agent, regardless of the reports of others.

A natural and widely adopted approach in practice is to interact with agents through sequential posted prices. Upon an agent’s arrival, a posted-price mechanism offers a payment \(p_a\) that may depend only on the history observed so far. The agent accepts the offer if and only if it exceeds their private cost, revealing minimal information about their valuation. Posted-price mechanisms are simple to understand and implement, universally truthful by construction and preserve agent privacy making them particularly attractive in applications such as crowdsourcing and online labor markets.

Since their introduction, procurement auctions have attracted substantial attention from the computer science and operations research communities due to both their practical relevance and theoretical richness. This interest has led to the study of numerous variants of the problem, distinguished by whether decisions are made offline or online, by the class of valuation functions considered, and by the interaction protocol between the employer and the agents. Below, we highlight the lines of work most closely related to our results; a broader discussion of additional variants and related progress is deferred to Section~\ref{sec: related work}.

Perhaps the closest line of work to ours is the study of posted-price procurement auctions with monotone submodular valuations, initiated by \citet{Bada2012}. They introduced the online version of the problem under secretary arrivals and presented a constant-competitive mechanism in the direct-revelation model, along with an \(O(\log n)\)-competitive posted-price mechanism for monotone submodular valuations. More recently, \citet{CFPT2025} obtained the first constant-competitive posted-price mechanism for this setting -- albeit with a very large constant -- thereby demonstrating that posted-price mechanisms can match direct-revelation mechanisms in terms of asymptotic performance. Our results for monotone submodular valuations significantly strengthen this line of work by improving the competitive ratio by a factor of $200$.

Another closely related line of research concerns online procurement auctions with non-monotone submodular valuations. \citet{AmanatidisKS22} established constant-competitive guarantees in the online setting under the direct-revelation protocol. In contrast, our work shows that comparable guarantees can be achieved even under the more restrictive posted-price model, thereby extending the scope of constant-competitive online procurement mechanisms beyond monotone objectives.
Our results provide a definitive positive answer to the motivating question of \citet{Bada2012} about the relative power of posted-price mechanisms for online budget feasible auctions compared against their direct revelation counterparts.

Finally, we complement our positive results with a logarithmic lower bound for XOS valuations, which applies to any online budget-feasible mechanism and constitutes, to the best of our understanding, the first impossibility result of this kind for online procurement auctions. Interestingly, these lower bounds establish a gap between online and offline budget-feasible procurement auctions: for XOS valuations, constant-factor approximation mechanisms are known offline, whereas our results show that such guarantees are impossible online.

\subsection{Technical Overview and Contribution}

In this work, we introduce \emph{Repeated Descent} (\RED), a deterministic routine tailored to the large-market regime of online procurement auctions with secretary arrival order. We use \RED as a core building block to obtain improved guarantees for posted-price online procurement auctions with monotone submodular valuations, and to design the first constant-competitive posted-price mechanism for online procurement with non-monotone submodular valuations. Finally, we establish a logarithmic lower bound for XOS valuations that applies to \emph{any} online budget-feasible mechanism.

Before describing our mechanism, we briefly recall a central primitive from the literature on budget-feasible procurement: \emph{linear pricing}. Given a target value $\hat t$, referred to throughout as a \emph{threshold} and intended to approximate the offline optimum $\OPT$, the mechanism offers an arriving agent $a$ a payment proportional to its marginal contribution,
\(
p_a \;=\; \frac{B}{\hat t}\cdot v(a\mid T),
\)
where $T$ denotes the set of previously accepted agents and $v(a\mid T):=v(T\cup\{a\})-v(T)$ is the marginal value of $a$ with respect to $T$. A well-established principle in budget-feasible procurement is that linear pricing is effective whenever it operates at the correct scale. In particular, if the mechanism spends a constant fraction of the budget while posting prices computed using thresholds that are within a constant factor of $\OPT$, then it necessarily recovers a constant fraction of the optimal value. On the other hand, if the threshold is significantly lower (higher) than $\OPT$, then the posted prices are too generous (too stingy), and the mechanism cannot achieve high value.

Building on this principle, prior approaches have focused on identifying an appropriate pricing scale, either by explicitly guessing a threshold~\cite{Bada2012} or by setting up a convoluted, multi-stage search to find a good threshold with positive probability~\cite{CFPT2025}. 
Our approach departs from this paradigm.

 \RED induces dynamics that steer the mechanism toward an appropriate threshold without ever estimating it explicitly. It begins at the maximum feasible threshold and repeatedly halves it (i.e., repeatedly descends over candidate thresholds) until a prescribed value cutoff is met; once sufficient value has been accumulated, it resets the threshold to its maximum level to prevent prolonged overspending. A key structural feature of \RED is that it enforces budget feasibility deterministically: by design, the routine expends agents faster than it expends budget, so the supply of agents is exhausted before the budget can be.
 

 These dynamics yield a win--win structure. Either the cutoff is reached frequently, in which case each success certifies that the mechanism has made steady progress by accumulating a fixed amount of value. Or successes are infrequent, in which case the threshold can only drift downward and therefore remains low for long stretches of the execution. 

In the latter regime, \RED makes many offers using thresholds below the desired scale, namely below a constant-factor approximation of $\OPT$, while still never risking budget exhaustion. During these stretches, the prices posted by \RED dominate the linear prices corresponding to any threshold within a constant factor of $\OPT$ (with respect to the same set of previously accepted agents). As a consequence, every agent who would accept such a linear-price offer is also accepted by \RED.

Therefore, in this regime it suffices to show that applying linear pricing at a threshold that is a constant-factor approximation of $\OPT$ extracts a constant fraction of the optimal value from the agents encountered during these rounds. Informally, it is standard in the budget-feasible procurement literature that linear pricing at an appropriate scale can extract a constant fraction of the value \emph{present} among the agents to which it is applied, provided the mechanism remains budget-feasible. Thus, after establishing that \RED spends a substantial portion of the execution posting prices that dominate those corresponding to thresholds near the desired scale, the main remaining challenge is to certify that these time periods indeed contain substantial value.

We address this using the random-order arrival model together with the large-market assumption and a deterministic (permutation independent) partitioning of the input. Intuitively, large markets preclude instances in which $\OPT$ is concentrated in a few exceptional agents, and random-order arrivals prevent the adversary from systematically placing all high-value agents outside the portions of the sequence where \RED posts low thresholds. To quantify the value in these portions, we partition the second phase into fixed-size \emph{rounds} matched to the threshold scale and use concentration of measure to show that many rounds are \emph{dense}, i.e., they contain significant value at the scale of $v_{\max}$. Separately, a deterministic counting argument lower bounds the number of rounds on which \RED posts low thresholds whenever successful descents are infrequent. Putting everything together implies that \RED operates at low thresholds on many dense rounds, and the standard linear-pricing guarantee then yields constant-factor value recovery.

We first instantiate the above framework for monotone submodular valuation functions. Our main result for this setting is the following, which improves on the best known posted-price competitive ratio of \citet{CFPT2025} by a factor of at least $200$.

\begin{theorem}
\label{thm:monotone informal}
There exists a universally truthful, budget-feasible posted-price mechanism for online procurement auctions with secretary arrivals and monotone submodular valuations that is $1046$-competitive.
\end{theorem}

We next extend the framework to non-monotone submodular valuations. The high-level structure of the analysis remains the same, but one key implication used in the monotone case no longer holds. 
Even in the monotone case, early acceptances can reduce the marginal contributions of later agents, so that even agents that would accept prices at the desired threshold may reject due to diminished marginals. However, using monotonicity, we bound this loss by charging it to the value already collected by the mechanism.

This charging argument fails in the non-monotone case. To address this difficulty, we leverage a powerful lemma due to \citet{Feige2011}, which shows that for any submodular valuation, constructing a random subset of a given set by independently including each element with probability $1/2$ yields an expected value of at least one quarter of the \emph{optimal} value attainable on that set. Since the optimal value over a set is monotone with respect to inclusion, this result allows us to recover value guarantees even when the mechanism accepts a superset of agents. By combining this randomized subsampling step with the structural guarantees provided by \RED, we are able to extend the previous logic to the non-monotone setting.


Our main result is the first competitive posted-price mechanism for non-monotone submodular valuations.

\begin{theorem}
\label{thm:non-monotone}
There exists a universally truthful, budget-feasible posted-price mechanism for online procurement auctions with secretary arrivals and non-monotone submodular valuations that is $181000$-competitive.
\end{theorem}

Finally, we complement our positive results with a logarithmic lower bound for online procurement under XOS valuations. Our construction is inspired by the framework of \citet{babaioff}. We partition the agents into $L$ disjoint groups, each of size $\ell$, and define an XOS valuation as the maximum of $L$ additive clauses, one per group. Each agent contributes only to the clause of its own group. We set the cost of every agent to $c(a)=B/\ell$, so any budget-feasible mechanism can hire at most $\ell$ agents in total.

The only remaining randomness lies in the additive coefficients. Independently for each agent, we set its coefficient to $1$ with probability $1/\ell$ and to $0$ otherwise. Equivalently, the valuation is drawn at random from a family of XOS instances with sparse additive clauses. In this distribution, the offline optimum corresponds to selecting the best group, i.e., the group containing the largest number of $1$-coefficients. When $L$ is sufficiently large, standard extreme-value bounds imply that the best group contains $\Theta(\ell)$ active agents in expectation, and hence $\E{}{\OPT}=\Theta(\ell)$.

In contrast, an online mechanism does not know future coefficients. Once it hires an agent from a group, the additional value it can hope to obtain from the remaining agents of that group is governed by a binomial random variable with expectation $<1$. Moreover, since the mechanism can hire at most $\ell$ agents overall, it can interact meaningfully with at most $\ell$ groups. As a result, the value of any online budget-feasible mechanism is upper bounded by the maximum of $\ell$ such binomial random variables, which is only $O(\log\ell)$ in expectation. Setting $\ell=\frac{\log n}{\log\log n}$ yields the following theorem.

\begin{theorem}
\label{thm:xos-lower-bound informal}
For procurement auctions with XOS valuation functions, no online budget-feasible mechanism can achieve a competitive ratio better than
$\Omega\!\left(\tfrac{\log n}{(\log\log n)^2}\right)$.
\end{theorem}

\subsection{Further Related Work}\label{sec: related work}

Ignoring strategic considerations, maximizing a monotone submodular function subject to a knapsack constraint is a classical NP-hard optimization problem. The greedy algorithm of \citet{NemhauserWF78} achieves an approximation ratio of $e/(e-1)$, which is optimal under standard complexity assumptions~\cite{KhullerMN99}. This guarantee extends to the knapsack-constrained setting~\cite{Sviri04}. More recently, \cite{BadaDO19} presented a $(9/8+\varepsilon)$-approximation using polynomially many demand queries.

\paragraph{Offline budget-feasible procurement.} \citet{Singer10,Singer13} initiated the study of budget-feasible procurement auctions from a mechanism design perspective and gave the first constant-factor approximation for monotone submodular valuations using truthful mechanisms. \citet{ChenGL2011} substantially improved these guarantees, achieving approximation ratios of $7.91$ (randomized) and $8.34$ (deterministic) for monotone submodular valuations, and $3$ (randomized) and $2+\sqrt{2}$ (deterministic) for additive valuations. They also established unconditional lower bounds of $1+\sqrt{2}$ (deterministic) and $2$ (randomized universally truthful) for additive valuations. For additive valuations, \citet{GravinJLZ2020} matched the optimal randomized approximation ratio of $2$ and gave a deterministic $3$-approximation. \citet{AnariGN2014} obtained optimal $e/(e-1)$-approximation mechanisms for additive valuations and improved ratios for monotone submodular valuations in large markets. For monotone submodular valuations, \citet{JalalyT2021} gave a randomized $5$-approximation mechanism and improved bounds assuming access to exact value maximization. Subsequent work refined these guarantees further: \citet{BalkanskiGGST2022} introduced polynomial-time clock auctions with a $4.75$-approximation, and \citet{HanWHC23} achieved the currently best known ratios of $4.45$ (deterministic) and $4.3$ (randomized).

Polynomial-time constant-factor approximations are also known for non-monotone submodular valuations. \citet{AmanatidisKS22} and \citet{HuangHC023} obtained constant-competitive mechanisms, with the best known guarantees being $64$ for deterministic and $12$ for randomized mechanisms, both achieved via clock auctions~\cite{BalkanskiGGST2022,HanWHC23}. For XOS valuations, constant-factor approximation mechanisms are known in the offline setting, albeit with non-polynomial running time~\cite{BeiCGL2012,AmanatidisBM17}. For subadditive valuations, \citet{BalkanskiGGST2022} gave an $O(\log n/\log\log n)$-approximation via clock auctions, matching earlier randomized guarantees~\cite{BeiCGL2012} and improving on the $O(\log^3 n)$ deterministic bound of \citet{DobzinskiPS11}. Recently, \citet{swamy-subadditive} improved upon this providing an $O(\log\log n)$ approximate mechanism. We refer the reader to the survey of \citet{LCLW2024_survey} for a comprehensive overview.

\paragraph{Online budget-feasible procurement.}
The online version of budget-feasible procurement auctions, where agents arrive sequentially in random order and decisions are irrevocable, was introduced by \citet{SingerM13} and \citet{Bada2012}. In the bidding model, \citet{AmanatidisKS22} obtained constant-competitive randomized universally truthful mechanisms for both monotone and non-monotone submodular valuations. Online budget-feasible procurement is closely related to the submodular knapsack secretary problem~\cite{BateniHZ13,KesselheimT17}, where costs are revealed upon arrival.
\citet{FeldmanNS11} presented a randomized $20e$-approximation for this problem.

\paragraph{Beyond bidding.} Posted-price mechanisms are particularly appealing due to their simplicity and obvious truthfulness, but they are significantly more restrictive than bidding-based protocols. Clock auctions~\cite{BalkanskiGGST2022,HanWHC23} avoid explicit bidding but are not online, as decisions regarding agents' acceptance in the solution are revocable. 
%
%
In the online setting with irrevocable decisions upon arrival, the only known posted-price mechanisms for monotone submodular valuations prior to our work are the $O(\log n)$-competitive mechanism of ~\cite[Section~4]{Bada2012} and the constant competitive mechanism of \cite{CFPT2025} discussed before.

\section{Preliminaries}
\label{sec:prelim}

We study posted-price online procurement auctions under a hard budget constraint. There is a set $N=\{1,\dots,n\}$ of agents. Each agent $a\in N$ has a privately known cost $c(a)\in\mathbb{R}_{\ge 0}$, representing the minimum payment at which they are willing to provide their service. The employer has budget $B\in\mathbb{R}_{\ge 0}$ and preferences given by a valuation $v:2^N\to\mathbb{R}_{\ge 0}$. The goal is to select a subset $T\subseteq N$ so as to maximize $v(T)$ subject to the budget constraint on total payments. 

We work in the random-order (secretary) model: an adversary fixes the instance (the costs and the valuation), and then agents arrive in a uniformly random permutation. We assume that the mechanism has access to function $v$ in the form of unlimited value queries on the set of revealed agents. Upon each arrival, the mechanism must make irrevocable decisions. We assume $n$ is known in advance. For any $A\subseteq N$, we write $c(A):=\sum_{a\in A}c(a)$.

\paragraph{Mechanisms and desiderata.}
An (offline) direct-revelation mechanism is specified by an allocation rule $A:\mathbb{R}_{\ge 0}^n\to 2^N$ and a payment rule $p:\mathbb{R}_{\ge 0}^n\to\mathbb{R}_{\ge 0}^n$. Given reported costs $\mathbf{b}=(b_1,\dots,b_n)$, the mechanism selects $A(\mathbf{b})\subseteq N$ and assigns payments $p(\mathbf{b})=(p_1(\mathbf{b}),\dots,p_n(\mathbf{b}))$. Agents have quasi-linear utilities: agent $a$'s utility is
\[
u_a(\mathbf{b}) :=
\begin{cases}
p_a(\mathbf{b})-c(a) & \text{if } a\in A(\mathbf{b}),\\
0 & \text{otherwise,}
\end{cases}
\]
and we adopt the convention that $p_a(\mathbf{b})=0$ whenever $a\notin A(\mathbf{b})$.
A mechanism is \emph{budget-feasible} if $\sum_{a\in N}p_a(\mathbf{b})\le B$ for all $\mathbf{b}$, and \emph{individually rational} if $u_a(\mathbf{b})\ge 0$ for all agents $a$ and all $\mathbf{b}$.
It is \emph{truthful} if reporting the true cost is a dominant strategy: for every agent $a$, every report $b\ge 0$, and every $\mathbf{b}_{-a}$,
$u_a((c(a),\mathbf{b}_{-a}))\ge u_a((b,\mathbf{b}_{-a}))$.
A randomized mechanism is a distribution over deterministic mechanisms; we require budget-feasibility and individual rationality with certainty and \emph{universal truthfulness}.

\paragraph{Posted-price mechanisms.}
We restrict attention to posted-price mechanisms. When an agent $a$ arrives, the mechanism posts a single take-it-or-leave-it offer $p_a$, which may depend on the history but not on future arrivals. Agent $a$ accepts iff $p_a\ge c(a)$; if they accept, their services are available to the employer and they are paid $p_a$, otherwise they leave irrevocably. Posted-price mechanisms are truthful and individually rational by construction, and they are budget-feasible provided each posted price does not exceed the remaining budget.

\paragraph{Valuation classes and oracle access.}
We assume $v$ is normalized, i.e., $v(\emptyset)=0$. For convenience, we write $v(a)$ for $v(\{a\})$ and define the marginal value $v(a\mid S):=v(S\cup\{a\})-v(S)$. The valuation $v$ is \emph{monotone} if $v(S)\le v(T)$ for all $S\subseteq T$, and \emph{submodular} if it has diminishing returns, i.e., $v(a\mid T)\le v(a\mid S)$ for all $S\subseteq T$ and $a\notin T$.
A valuation is \emph{XOS} (fractionally subadditive) if there exists a collection of additive functions $\{\mu_k\}_k$ such that for every $S\subseteq N$,
$v(S)=\max_k \sum_{a\in S}\mu_k(a)$.
We access $v$ via value queries, which return $v(S)$ for any $S\subseteq N$. Let $v_{\max}:=\max_{a\in N} v(a)$ denote the maximum singleton value.

\paragraph{Linear prices.}
A standard reference point in budget-feasible procurement is \emph{linear pricing}, which posts prices proportional to marginal contributions. Given a target value $\hat t$ intended to approximate the offline optimum, one offers an arriving agent $a$ the price
\(
p_a=(B / \hat{t})\cdot v(a\mid T),
\)
where $T$ is the set of already accepted agents.

\paragraph{Competitive ratio.}
Let $\OPT$ denote the value of the optimal offline solution that observes all agents and their costs in advance and selects a set $C^*\subseteq N$ with $c(C^*)\le B$. Let $\ALG$ denote the (random) value achieved by an online mechanism; expectations are over the random arrival order and any internal randomness. A mechanism is $\alpha$-competitive if $\mathbb{E}[\ALG]\ge \OPT/\alpha$ for all instances.

\section{Monotone Submodular Valuations}

This section introduces and analyzes \RED, the routine at the core of both our monotone and non-monotone large-market mechanisms. Given a budget \(B\), a set \(N_{\text{rem}}\) of remaining agents, and a scale parameter \(\vmax\) (an estimate of the maximum singleton value), \RED computes an adaptive sequence of thresholds. The mechanism uses these thresholds to post linear prices to agents.
We then restrict attention to monotone submodular valuations and show that \RED achieves a constant competitive ratio on large-market instances $(\OPT \geq 8000\cdot \vmax)$. Finally, we randomize between \RED and two posted-price mechanisms tailored to small and medium markets, obtaining \textsc{PostedPrice} with the following guarantee.

\begin{theorem}\label{thm:monotone}
\textsc{PostedPrice} is a universally truthful, $1046$-competitive posted-price mechanism for online budget-feasible procurement auctions with secretary arrivals and monotone submodular valuations.
\end{theorem}

\begin{algorithm}[t]
\caption{\LM}\label{alg:LMMECH}
\DontPrintSemicolon
\KwInput{Set $N$ of agents, budget $B$}
$\vmax \gets \textsc{LearningMaxValue}([n/2])$\;
$T \gets \RED([n/2+1:n], B, \vmax)$\;
\KwRet{$T$}\;
\end{algorithm}

\subsection{Repeated Descent (\RED)}

\begin{algorithm}[t]
\caption{\RED}\label{alg:RED}
\SetAlgoNoEnd
\DontPrintSemicolon
\KwInput{Remaining agents $N_{\text{rem}}$, budget $B$, maximum value $\vmax$}
$T \gets \varnothing$, $\Brem \gets B$\;
\While{$|N_{\text{rem}}| \ge 50$}{
    $t \gets n \cdot \vmax$, $V \gets 0$, $\tau \gets 50 \frac{\vmax}{t} n$\;
    \While{$V < 49 \, \vmax$}{
        \If{$\tau \le 0$}{
            $t \gets t/2$, $\tau \gets 50 \frac{\vmax}{t} n$\;
            \If{$|N_{\text{rem}}| < \tau$}{
            \KwRet{$T$}\;
        }
        }
        Let $a$ be the next arriving agent\;
        \If{$v(a \mid T) \cdot B/t \leq \Brem$}{
        Offer price $p = v(a\mid T) \cdot B / t$ to agent $a$\;
        \If{$p \geq c(a)$}{
            $V \gets V + v(a\mid T)$, $T \gets T \cup \{a\}$, $\Brem \gets \Brem - p$\;
        }
    }
    $N_{\text{rem}} \gets N_{\text{rem}} \setminus \{a\}$, $\tau \gets \tau - 1$\;
    }

}
\KwRet{$T$}\;
\end{algorithm}

\RED merely loops: it starts a descent whenever enough agents remain to do so. The descent subroutine is described below.

\paragraph{The descent subroutine.}
A descent maintains a threshold \(t\) and a running total \(V\) of value collected \emph{during the current descent}. It initializes \(t \gets n\cdot \vmax\) and \(V \gets 0\), and proceeds in stages indexed by the current threshold \(t\).
At threshold \(t\), the mechanism processes a block of $\tau(t) := 50\cdot \frac{\vmax}{t}\cdot n $
consecutive arriving agents. For each such agent \(a\), it posts the linear price $p(a) = \frac{B}{t}\cdot v(a\mid T),$
where \(T\) denotes the current solution set at that time. If \(a\) accepts, then \(a\) is added to \(T\) and \(V\) increases by \(v(a\mid T)\). If at any point \(V \ge 49\cdot \vmax\), the descent terminates successfully.
If the block ends without reaching this target, the threshold is halved (\(t \gets t/2\)), the block length is updated to \(\tau(t)\), and the descent continues with the next agents. A descent ends either upon success or when fewer than \(\tau(t)\) agents remain to execute the next block. The term ``descent'' reflects that the subroutine searches downward over thresholds, starting from \(n\cdot \vmax\) and decreasing \(t\) geometrically from block to block.


\subsection{Budget Feasibility}

A crucial property of \RED is that it never exhausts its budget.

\begin{lemma}\label{lemma:budget}
By design, when \RED is applied to $n/2$ agents with budget $B$ and a correct estimate of $\vmax$, the total payment made by the routine is at most $B$ and the budget unavailability condition is never enforced.
\end{lemma}

\begin{proof}
Consider an arbitrary descent, and let $t$ denote the smallest threshold reached during that descent. We first show that any descent processes at least $R_t = 25 \cdot \frac{\vmax}{t} \cdot n$ agents before terminating. We consider the following two cases.

\begin{enumerate}
    \item The descent terminates at the initial threshold $t = n\cdot \vmax$. In this case, termination requires that the accumulated value reaches $49\cdot \vmax$, which can only occur after processing at least $49$ agents. Since $R_t = 25 \cdot \frac{\vmax}{n\vmax}\cdot n = 25$, the bound holds trivially.
    \item The descent reaches a smaller threshold $t < n\cdot \vmax$. To reach threshold $t$, the previous block (at threshold $2t$) must have processed $50\cdot \frac{\vmax}{2t}\cdot n \;>\; R_t$ agents without reaching the termination condition. Hence, the total number of agents processed during the descent is at least $R_t$.
\end{enumerate}

We now bound the budget spent in a single descent. The maximum budget consumption occurs if all accepted agents were offered prices corresponding to the minimum threshold $t$, and the final accepted agent contributes value $\vmax$. In this case, the total payment made during the descent is at most
\(
B_t = 50 \cdot \frac{\vmax}{t} \cdot B.
\)
Therefore, the average budget spent per processed agent in this descent is at most
\(
\frac{B_t}{R_t}
=
\frac{50\cdot B}{25\cdot n}.
\)
Since \RED is applied only to the final $n/2$ agents, the total expenditure across all descents is at most
\(
\frac{50\cdot B}{25\cdot n}\cdot \frac{n}{2}
=
B,
\)
which completes the proof.
\end{proof}

\subsection{Analysis Outline}

Our goal is to recover a constant fraction of the offline optimum $\OPT$. Each successful descent certifies progress: it terminates only after \RED\ has accumulated $\Theta(\vmax)$ marginal value during that descent. Hence, if \RED\ completes $\Theta(\OPT/\vmax)$ successful descents, then it already obtains $\Omega(\OPT)$ value.

On the other hand, if the number of successful descents is smaller, we argue that \RED\ must have made offers to a large number of agents using thresholds below $t'$, where $t'=\Theta(\OPT)$. To translate this argument into value gathered by \RED, we proceed as follows:
\begin{enumerate}
    \item We partition the input sequence into fixed-size contiguous blocks, called rounds.
    \item We define a set $C^*$ of agents with a high value-to-cost ratio and argue that, under the large-market and secretary-order assumptions, many rounds are \emph{dense}, i.e., they contain sufficient value contributed by the high-return agents in $C^*$.
    \item Finally, we prove that if a threshold lower than $t'$ is used in many dense rounds, then \RED\ must obtain $\Theta(\OPT)$ value overall.
\end{enumerate}

\subsection{Rounds, Benchmark Set and Density}\label{sub:rounds-conc}


We have already seen that \RED processes the stream in consecutive blocks.
In this subsection, we show that once a block is large enough (on the order of
$\frac{\vmax}{\OPT}\,n$ arrivals) it contains a nontrivial amount of value, namely $\Theta(\vmax)$, with constant probability under random-order arrivals.
This is the main structural input for the analysis; in later subsections, we show how to extract a constant fraction of the value present in such blocks.

Fix a threshold of interest \(t'\), for monotone submodular valuations $t'$ is the unique threshold of the form $n\cdot \vmax / 2^r$, $r \in \mathbb{N}$ such that $\OPT/16 \leq t' < \OPT/8$. We partition the arrival sequence as follows:
\begin{itemize}
    \item The first \(n/2\) agents belong to the sampling phase and are ignored in the remainder of the analysis.
    \item The remaining \(n/2\) agents are divided into \emph{rounds}. A \emph{round} consists of $\ell=50\cdot \frac{\vmax}{t'}\cdot n$ consecutive agents.
    Any leftover agents that do not complete a full round are disregarded.
\end{itemize}
For each round index \(j\), let \(A_j\subseteq N\) denote the (random) set of agents that belong in round \(j\).

To quantify the amount of value present in each round, we introduce a benchmark set $C^*$. We remark that the set $C^*$ is used solely for analysis and does not depend on the execution of \RED. 

\begin{lemma}[\cite{Bada2012}]\label{lemma:Svalue}
Let $C^*$ be the set of accepting agents obtained by applying threshold $\hat{t} = \OPT/2$ to the entire input under ordering $\pi_0$. Then
\[
\frac{\OPT - v_{\max}}{2} \le v(C^*) \le \frac{\OPT}{2}.
\]
\end{lemma}

Let $(a_1,a_2,\dots,a_\nu)$ be the agents in $C^*$, ordered by their appearance in $\pi_0$. For each $a_i\in C^*$, define its marginal contribution
\[
w(a_i) := v(\{a_1,\dots,a_i\}) - v(\{a_1,\dots,a_{i-1}\}).
\]
By definition, $\sum_{a\in C^*} w(a) = v(C^*)$.

Informally, we lower bound the value present in a round by the marginal values of the $C^*$-agents that appear in that round. Since $C^*$ captures a constant fraction of $\OPT$, random-order arrivals and the large-market assumption imply that, for each fixed round of size $\ell$, this benchmark contribution is $\Theta(\vmax)$ with constant probability.

Fix a round \(j\). For each position \(i\in[\ell]\) within round \(j\), define the random variable
\begin{equation}\label{eq:Xij}
X_i^j \;=\;
\begin{cases}
w(a) & \text{if the agent occupying position \(i\) of round \(j\) is some } a\in C^*\setminus\{a_{\max}\}\text{\footnotemark},\\
0 & \text{otherwise.}
\end{cases}
\end{equation}
\footnotetext{We exclude \(a_{\max}\) because we will condition on the event that an agent of value \(\vmax\) appears in the first half.}
Below we prove that 
\(
X^j := \sum_{i=1}^{\ell} X_i^j
\) provides a lower bound for the value that high-return agents of $C^*$ contribute in each round.

\begin{lemma}\label{lem:round-submod}
For every round \(j\), we have
\(
v(C^*\cap A_j) \ge X^j.
\)
\end{lemma}

\begin{proof}
Let \(S_j:=C^* \cap A_j\) be the benchmark agents that appear in round \(j\). By definition, \(X^j=\sum_{a\in S_j} w(a)\), since each \(a\in C^*\) contributes exactly \(w(a)\) if it appears in the round and \(0\) otherwise.

Now consider the ordered list \((a_1,\dots,a_k)\) defining the marginals \(w(\cdot)\). For any subset \(S\subseteq C^*\), submodularity implies that the marginal of \(a_i\) with respect to a \emph{subset} of \(\{a_1,\dots,a_{i-1}\}\) is at least its marginal with respect to \(\{a_1,\dots,a_{i-1}\}\). Hence, if we build \(S\) in the order induced by \(\pi_0\), the contribution of each included element \(a_i\in S\) is at least \(w(a_i)\). Summing over \(a_i\in S\) yields \(v(S)\ge \sum_{a_i\in S} w(a_i)\). Applying this to \(S=S_j\) gives \(v(C^* \cap A_j)\ge X^j\).
\end{proof}

\paragraph{Concentration.}

We proceed to define high-value rounds, called "dense" and argue that a constant fraction of the rounds defined are dense with probability at least a half.

\begin{definition}[Dense round]\label{def:dense-round}
Let \(W_j := C^* \cap A_j\) be the benchmark agents appearing in round \(j\). In the monotone submodular case, we call round \(j\) \emph{dense} if
\[
\sum_{a\in W_j} w(a) \;\ge\; 160\,v_{\max}.
\]
\end{definition}


The variables \(\{X_i^j\}\) are not independent but are negatively associated, as they arise from sampling without replacement (a permutation distribution). By \cite{negass} and \cite{chen2017bernstein}, Bernstein’s inequality applies under negative association. 

\begin{theorem}[Bernstein's Inequality]\label{thm:bern-ineq}
Let \(X_1,\dots,X_\ell\) be independent (or negatively associated) random variables with \(\E{}{X_i}=\mu_i\) and \(X_i-\mu_i\le b\).
Let \(X=\sum_{i=1}^\ell X_i\), \(\mu=\E{}{X}\), and \(V=\sum_{i=1}^\ell \E{}{(X_i-\mu_i)^2}\).
Then for any \(t>0\),
\[
\Prob{X \le \mu - t} \;\le\; \exp\!\left(- \frac{t^2}{2V+\frac{2bt}{3}}\right).
\]
\end{theorem}

Applying Bernstein’s inequality yields the following bound; the calculation is deferred to Appendix~\ref{bern}.

\begin{lemma}\label{lemma:Bern-appl}
Let \(\mathcal{E}\) denote the event that an agent of value \(\vmax\) appears in the first half of the input sequence.
For every round \(j\),
\[
\Prob{ X^j \le 160\,v_{\max} \,\middle|\, \mathcal{E}} \le 0.1.\]
\end{lemma}

\begin{corollary}
Let $R$ be the number of rounds and let $D$ be the number of dense rounds, it holds
$$ \E{}{D \mid \mathcal{E}} \geq 0.9 \cdot R.$$
\end{corollary} 

\begin{corollary}\label{cor:number of dense}
    Let $R$ be the number of rounds and let $D$ be the number of dense rounds, it holds 
    $$\Prob{D \ge 0.8\,R \,\middle|\, \mathcal{E}} \ge \tfrac{1}{2}.$$
\end{corollary}

We analyze the performance of \RED in instances where there are $D \ge 0.8\,R$ rounds. We denote this event by $\mathcal{F}$ and will condition on it from here on out.


\subsection{Counting Lemma}

Throughout this subsection, we assume that \RED uses the correct scale \(\vmax\) (event $\mathcal{E})$. Fix a threshold cut \(t'\) and recall that $\ell := 50\cdot \frac{\vmax}{t'}\cdot n$. We begin with a fundamental statement about \RED's threshold usage. 

\begin{lemma}\label{lem:offer-mono}
Consider any time step during the execution of \RED. If no descent has succeeded among the previous \(\ell\) agents, then the current offer is made using a threshold \(t \leq t'\).
\end{lemma}

\begin{proof}
Fix a time step, and let \(s\) be the time of the most recent successful descent (or the start of the procedure if no success has occurred yet). Consider the descent that is currently running (i.e., the one that started right after time \(s\)).

Within a descent, thresholds are halved from \(n\vmax\) downwards, and at threshold value \(t\) the mechanism spends a block of
\(50\cdot \frac{\vmax}{t}\cdot n\) agents before halving again (unless the descent succeeds earlier).
Therefore, the total number of agents that can be processed \emph{before the threshold drops to \(t'\)} is upper bounded by the sum of block lengths over thresholds \(t > t'\).
Since block lengths scale as \(\propto 1/t\) and the thresholds form a geometric sequence, this sum is a geometric series:
\[
\sum_{\text{levels }t > t'} 50\cdot \frac{\vmax}{t}\cdot n
\;\le\;
50\cdot \frac{\vmax}{t'}\cdot n \cdot \sum_{r=1}^{\infty} 2^{-r}
\;=\;
\ell.
\]
Consequently, if more than \(\ell\) agents have been processed since time \(s\) without a success, then the descent must have reached a threshold \(t\leq t'\). This concludes the proof of the lemma.
\end{proof}

The next lemma turns this per-step monotonicity into a round-counting statement.

\begin{lemma}\label{lem:counting}
Let \(R\) be the number of rounds in the second half of the input sequence (each of length \(\ell\)), and let \(d\) be the number of completed (successful) descents. Then \RED makes offers using thresholds equal or smaller than \(t'\) on at least \(\max\{0,\; R - 2d - 1\}\) rounds.
\end{lemma}

\begin{proof}
Call a round \emph{high-threshold} if \RED makes at least one offer in that round using a threshold \(t > t'\).
By Lemma~\ref{lem:offer-mono}, every such offer must occur within \(\ell\) processed agents after a successful descent (or within the first \(\ell\) processed agents of the entire execution, before any success happens).

Thus, there are only two ways a round can be high-threshold:
(i) it intersects the first \(\ell\) processed agents of the execution (at most one round), or (ii) it lies within \(\ell\) processed agents after some successful descent. Each successful descent can keep at most two rounds from being entirely low-threshold: the round in which the success occurs, and (possibly) the following round, since the next descent restarts at \(n\vmax\) and needs up to \(\ell\) agents before it can drop to \(t'\).
Therefore, the total number of high-threshold rounds is at most \(2d+1\).

All remaining rounds are low-threshold, i.e., they contain only offers with thresholds \(\leq t'\).
Hence the number of such rounds is at least \(R-(2d+1)\), and truncating at \(0\) yields the stated bound.
\end{proof}



\subsection{Combining Counting and Density}

We now combine the deterministic counting guarantee (Lemma~\ref{lem:counting}) with the event \(\mathcal{F}\) established in Section~\ref{sub:rounds-conc}, which asserts that many rounds are dense. This yields a dichotomy: either many descents succeed, or \RED operates at thresholds \(\leq t'\) on many dense rounds.

\begin{lemma}\label{lem:dichotomy}
Conditioning on event \(\mathcal{F}\), at least one of the following holds:
\begin{enumerate}
    \item \RED completes at least \(R/4\) successful descents.
    \item \RED makes offers using thresholds \(t \leq t'\) on at least \(0.1R\) dense rounds.
\end{enumerate}
\end{lemma}

\begin{proof}
Assume (1) does not hold. Let \(d\) be the number of successful descents, so \(d < 0.25R\).
By Lemma~\ref{lem:counting}, the number of rounds on which \RED uses only thresholds \(\leq t'\) is at least
\[
R-(2d+1)\;>\; R-\Bigl(2\cdot 0.25R+1\Bigr)\;=\;0.5R-1.
\]
For \(R\ge 5\) (which is ensured by the large-market regime), we have \(0.5R-1 \ge 0.3R\).
So \RED uses thresholds \(\leq t'\) on at least \(0.3 R\) rounds. Conditioning on \(\mathcal{F}\), at most \(0.2R\) rounds are not dense. Therefore, among these at least \(0.3R\) low-threshold rounds, at least
\(
0.3R - 0.2R \;=\; 0.1R
\)
are dense, proving (2).
\end{proof}

\subsection{Value Recovery}

We conclude by analyzing the two cases in Lemma~\ref{lem:dichotomy} and lower bounding the value obtained by \RED in each case. Let $A$ denote the event that \RED makes offers using thresholds $\leq t'$ on at least $0.1R$ dense rounds. We consider two cases:

\textbf{Case 1: event \emph{A}.} Let $Q$ be the set of agents in dense rounds where \RED makes offers using $t \leq t'$ and let $Q' = Q \cap C^*$. Recall that $\ell = 50 \; \frac{\vmax}{t'} \; n $ and $\OPT/16 \leq t' < \OPT/8$. Finally, let $T$ be the solution at the end of the run. By the definition of density and submodularity, it holds 
$$v(Q') \geq 0.1R \cdot 160 \vmax \geq 0.1 (\frac{5}{6} \frac{n}{2 \ell}) \cdot 160 \vmax \geq \OPT/120$$ 
where the second inequality uses $R \geq \frac{5}{6}\frac{n}{2\ell}$ that holds by the large market assumption.  Using submodularity and monotonicity we get:
$$ v(Q') - v(T) \leq v(Q' \cup T) - v(T) \leq \sum_{a \in Q' \setminus T} v(a \mid T)$$
We know that all these agents rejected offers with threshold at most $\OPT/8$ thus:
$$ \sum_{a \in Q' \setminus T} v(a \mid T) \leq \frac{\OPT}{8B} \sum_{a \in Q' \setminus T} c(a) $$

Now recall that due to the definition of $C^*$ there exists an ordering $\pi_0$ with which the agents in $Q'$ are purchasable with threshold $\OPT/2$. Therefore, if we assume that indices in $Q'\setminus T = \{a_1, \ldots, a_k\}$ are ordered with respect to $\pi_0$ it holds
$$ c(a_i) \leq \frac{2B}{\OPT} v(a_i \mid \{a_1, \ldots, a_{i-1}\})$$

Summing over all agents in $Q' \setminus T$:
$$ \sum_{i = 1}^k c(a_i) \leq \frac{2B}{\OPT} \sum_{i = 1}^k v(a_i \mid \{a_1, \ldots, a_{i-1}\}) \leq \frac{2B}{\OPT} v(Q' \setminus T) \leq \frac{2B}{\OPT} v(Q')$$

Substituting back to the initial equation we get:
$$ v(Q') - v(T) \leq \frac{v(Q')}{4} \Rightarrow v(T) \geq \frac{3}{4} v(Q') \geq \frac{1}{160} \OPT. $$


\textbf{Case 2: negation of \emph{A}.} Conditioning on $\mathcal{E}$ and $\mathcal{F}$, Lemma~\ref{lem:dichotomy} implies that \RED completes at least \emph{R/4} successful descents. Since \RED completed $R/4$ descents, by the definition of completed descents it accumulated $\frac{R}{4} 49 \vmax$ value. Therefore:
$$\E{}{\ALG\mid \mathcal{E}, \mathcal{F}, \neg A} \geq \frac{R}{4} 49 \vmax \geq \frac{1}{4} \left(\frac56 \frac{n}{2\ell}\right) \cdot 49\vmax \geq \frac{\OPT}{160}$$

\noindent
Putting everything together
\begin{align*}
    \E{}{\ALG} & \geq \E{}{\ALG \mid \mathcal{E}, \mathcal{F}} \Prob{\mathcal{F} \mid \mathcal{E}} \Prob{\mathcal{E}}\\
    & \geq \frac{1}{4} \min \big\{\E{}{\ALG \mid \mathcal{E}, \mathcal{F}, A}, \E{}{\ALG \mid \mathcal{E}, \mathcal{F}, \neg A}\big\} \geq \frac{\OPT}{640}.
\end{align*}

\subsection{Small and Medium Markets}\label{sec:small}

In the previous subsections, we analyzed the performance of our large-market mechanism under the assumption that $\OPT \ge 8000 \cdot v_{\max}$. We now remove this assumption and present a complete mechanism that achieves a constant competitive ratio for \emph{all} instances. 

Our approach is to randomize over three truthful mechanisms, each tailored to a different market regime. The resulting mechanism remains universally truthful and budget-feasible.

\begin{algorithm}[ht]
\caption{\textsc{PostedPrice}}
\KwInput{Set of agents $N$, budget $B$}
with probability $80/1046$ run Dynkin's algorithm on $N$\;
with probability $326/1046$ run \textsc{MediumMarket} on $N$ with budget $B$\;
with probability $640/1046$ run \LM\ on $N$ with budget $B$\;
\end{algorithm}

When the instance is dominated by a single high-value agent, we rely on Dynkin's classical secretary algorithm.

\begin{theorem}\label{thm: Dynkin}
If $v_{\max} \ge \OPT/29$, then Dynkin's algorithm is $(29e)$-competitive.
In particular, it is $80$-competitive.
\end{theorem}

For intermediate regimes, we use a carefully tuned variant of the mechanism of~\cite{Bada2012}, called \textsc{MediumMarket} (formally presented in Section~\ref{sec:mediummarket} of the Appendix), which is optimized for the range $29\,v_{\max}\le \OPT \le 8000\,v_{\max}$. The analysis of this mechanism follows standard arguments and is deferred to the Appendix~\ref{sec:mediummarket}. 

\paragraph{Putting everything together.} We are ready to prove Theorem~\ref{thm:monotone}.

\begin{proof}
Universal truthfulness follows immediately, as \textsc{PostedPrice} is a randomization over universally truthful mechanisms. 

We bound the competitive ratio by considering three cases, depending on the value of $\OPT$. 
\begin{enumerate}
    \item $\OPT \le 29 \cdot v_{\max}$: with probability $80/1046$, \textsc{PostedPrice} executes Dynkin's algorithm, which achieves value at least $\OPT/80$ in expectation. Thus,
\[
\mathbb{E}[\ALG] \;\ge\; \frac{80}{1046}\cdot \frac{\OPT}{80} \;=\; \frac{\OPT}{1046}.
\]

    \item $29 \cdot v_{\max} < \OPT \le 8000 \cdot v_{\max}$: with probability $326/1046$, \textsc{PostedPrice} executes \textsc{MediumMarket}, which is $326$-competitive. Hence,
\[
\mathbb{E}[\ALG] \;\ge\; \frac{326}{1046}\cdot \frac{\OPT}{326} \;=\; \frac{\OPT}{1046}.
\]

    \item $\OPT \ge 8000 \cdot v_{\max}$: with probability $640/1046$, \textsc{PostedPrice} executes \LM, which is $640$-competitive. Therefore,
\[
\mathbb{E}[\ALG] \;\ge\; \frac{640}{1046}\cdot \frac{\OPT}{640} \;=\; \frac{\OPT}{1046}.
\]
\end{enumerate}

In all cases, $\mathbb{E}[\ALG] \ge \OPT/1046$, completing the proof.
\end{proof}

\section{Extension to non-Monotone Submodular Valuations}

In this section, we give the first constant-competitive posted-price mechanism for online procurement auctions with non-monotone submodular valuations. Since non-monotonicity raises a modeling choice that is irrelevant in the monotone case, we begin by specifying the commitment model we study.

A posted-price mechanism always pays an agent who accepts its offer, but in the non-monotone setting the employer may not want to actually use every purchased service. We therefore allow the mechanism to \emph{discard} hired agents (i.e. to pay them if they accept but to select not to use their services). To preserve the online nature of the problem, we require that this keep/discard decision is made irrevocably at the time of acceptance (i.e., immediately after the agent accepts the offer) \emph{without knowledge of future arrivals}. This commitment model differs fundamentally from clock auctions (see ~\cite{HanWHC23}), even when restricted to a single query per agent. Our guarantees obviously extend to the weaker model in which the employer may postpone all discarding decisions until the end of the sequence.

\subsection{Non-Monotone Repeated Descent (\textsc{NM-ReD})}

We present \textsc{NM-ReD}, a non-monotone analogue of \RED for large-market instances. At its core \textsc{NM-ReD} embeds the \textsc{GenSm-Online} mechanism of \citet{AmanatidisKS22} into our Repeated Descent framework. Specifically, it maintains two candidate solution tracks \(S_1,S_2\). Each accepted agent is assigned to the track on which it has larger marginal value (at the time of arrival), and the mechanism uses the same linear price rule as in \RED with respect to that track. To handle non-monotonicity, \textsc{NM-ReD} additionally forms random subsamples \(T_1\subseteq S_1\) and \(T_2\subseteq S_2\) by independently retaining each accepted agent with probability \(1/2\), and returns one of \(S_1,S_2,T_1,T_2\) at the end according to a fixed distribution.

\begin{algorithm}[ht]
\caption{\textsc{NM-ReD}}
\label{alg:nm-red}
\SetAlgoNoEnd
\DontPrintSemicolon
\KwInput{Remaining agents $N_{\text{rem}}$, budget $B$, maximum singleton value $\vmax$}

\textbf{Initialize:} $S_1 \gets \varnothing$, $S_2 \gets \varnothing$, $T_1 \gets \varnothing$, $T_2 \gets \varnothing$, $\Brem \gets B$\;
Sample the output set $T$ as follows:\;
\Indp
$T \gets S_1$ w.p.\ $3/10$, \quad $T \gets S_2$ w.p.\ $3/10$\;
$T \gets T_1$ w.p.\ $2/10$, \quad $T \gets T_2$ w.p.\ $2/10$\;
\Indm
\While{$|N_{\text{rem}}| \ge 50$}{
    $t \gets n\cdot \vmax$,$V \gets 0$,$\tau \gets 50\cdot \frac{\vmax}{t}\cdot n$\;

    \While{$V < 49\,\vmax$}{
        \If{$\tau \le 0$}{
            $t \gets t/2$,$\tau \gets 50\cdot \frac{\vmax}{t}\cdot n$\;
            \If{$|N_{\text{rem}}| < \tau$}{
            \KwRet{$T$}\;
        }
        }

        Let $a$ be the next arriving agent\;
        $\hat{j} \gets \arg\max_{j\in\{1,2\}} v(a \mid S_j)$, $p \gets \frac{B}{t}\cdot v(a \mid S_{\hat{j}})$\;
        \If{$p \leq \Brem$}{
        Offer price $p$ to agent $a$\;
                \If{$p \ge c(a)$}{
            $V \gets V + v(a \mid S_{\hat{j}})$, $S_{\hat{j}} \gets S_{\hat{j}} \cup \{a\}$, $\Brem \gets \Brem - p$\;

            With probability $1/2$: $T_{\hat{j}} \gets T_{\hat{j}} \cup \{a\}$\;
            Update $T$\;}}

        $N_{\text{rem}} \gets N_{\text{rem}}\setminus\{a\}$, $\tau \gets \tau - 1$\;
    }
}
\KwRet{$T$}\;
\end{algorithm}

As in the monotone case, our large-market mechanism, called \textsc{NM-LargeMarket}, uses the first half of the arrival sequence to estimate $\vmax$, and then runs the non-monotone variant of \RED on the remaining agents. In this subsection, we prove the following theorem.

\begin{theorem}\label{thm: non-monotone large}
If $\OPT \geq 50000 \cdot \vmax$ then \NLM\;is $44800$-competitive.
\end{theorem}

We begin by observing that budget feasibility carries over from Lemma \ref{lemma:budget}. Although the mechanism now maintains two parallel solution tracks, the budget analysis follows the same high-level structure as in the monotone case. 

Another structural observation is that the number of successful descents is bounded by the total value accumulated across the two tracks. By construction, each successful descent terminates after at least $49\vmax$ marginal value is added to $v(S_1) + v(S_2)$. Therefore, if $d$ is the number of successful descents, we have $v(S_1) + v(S_2) \geq 49\vmax \cdot d$.

As in the monotone case, we partition the second half of the input sequence into \emph{rounds} of length
\(
\ell := 50 \cdot \frac{\vmax}{t'} \cdot n,
\)
where $t'$ is the unique threshold of the form $n\cdot \vmax / 2^r$, $r \in \mathbb{N}$ that satisfies \(\OPT/200 \le t' < \OPT/100\).

To apply concentration, we compare against a fixed benchmark set. Let \(C^*=\{c_1,\dots,c_k\}\) be an \emph{inclusion-wise minimal} optimal solution, and fix an arbitrary ordering \(\pi_0\) of its elements. For each \(i\in[k]\), define the marginal contribution
\[
w(c_i)\;:=\; v(\{c_1,\dots,c_i\})-v(\{c_1,\dots,c_{i-1}\}).
\]
Then \(\sum_{i=1}^k w(c_i)=v(C^*)=\OPT\). Moreover, by minimality we have \(v(c_i\mid C^*\setminus\{c_i\})\ge 0\), and by submodularity
\(\{c_1,\dots,c_{i-1}\}\subseteq C^*\setminus\{c_i\}\) implies \(w(c_i)\ge 0\). Hence, for all \(i\in[k]\),
\[
0 \;\le\; w(c_i) \;\le\; \vmax.
\]

We define \(X^j\) exactly as in the monotone case, and update the density threshold to match the scale induced by \(t'\).

\begin{definition}[Dense round]\label{def:S-dense}
Let \(W_j := C^*\cap A_j\) be the benchmark agents appearing in round \(j\). We call round \(j\) \emph{dense} if
\[
\sum_{a\in W_j} w(a) \;\ge\; 2000\,\vmax.
\]
\end{definition}

Applying Bernstein's inequality (using negative association under random-order arrivals) yields the following bound.

\begin{lemma}\label{lem:nm-bern}
Let \(\mathcal{E}\) denote the event that an agent of value \(\vmax\) appears in the first half of the input sequence. For every round \(j\),
\[
\Prob{X^j \le 2000\,\vmax \,\middle|\, \mathcal{E}} \le 0.1.
\]
\end{lemma}

We next establish the analogue of Lemma~\ref{lem:dichotomy}. Its proof is identical since it is purely structural and does not rely on monotonicity, however we select a different combination of constants to better fit our proof.

\begin{lemma}\label{lem:dichotomy nm}
Conditioning on event \(\mathcal{F}\), at least one of the following holds:
\begin{enumerate}
    \item \textsc{NM-ReD} completes at least \(R/8\) successful descents.
    \item \textsc{NM-ReD} makes offers using thresholds \(t \leq t'\) on at least \(0.3R\) dense rounds.
\end{enumerate}
\end{lemma}

We proceed to prove Theorem \ref{thm: non-monotone large}.

\begin{proof}
We condition throughout on the event \(\mathcal{E}\) that an agent of value \(\vmax\) appears in the first half of the sequence, and on the event \(\mathcal{F}=\{D\ge 0.8R\}\) that at least \(0.8R\) rounds are dense. Let \(A\) denote the event that \textsc{NM-ReD} makes offers using thresholds \(\leq t'\) on at least \(0.3R\) dense rounds. We consider two cases.\\

\textbf{Case 1: event \emph{A}.}
Let \(\mathcal{G}\) be the set of dense rounds on which the mechanism makes offers with threshold \(t \leq t'\). For each \(j\in\mathcal{G}\), density implies \(\sum_{a\in W_j} w(a) \ge 2000\,\vmax\), where \(W_j=C^*\cap A_j\). Let \(C := \bigcup_{j\in\mathcal{G}} W_j\). Since rounds are disjoint, the sets \(\{W_j\}_{j\in\mathcal{G}}\) are disjoint, and therefore
\[
\sum_{a\in C} w(a)
\;\ge\;
2000\,\vmax\cdot |\mathcal{G}| 
\;\ge\; 5t',
\]
where the last inequality uses \(|\mathcal{G}|\ge 0.3R\).

We use a slightly generalized version of~\cite[Theorem 4.1]{AmanatidisKS22}; its proof is deferred to Appendix~\ref{sec: amanat}.

\begin{lemma}\label{lem:amanat}
Let \(C\) be any budget-feasible set of agents such that each \(a\in C\) was offered a price computed with some threshold \(t_a\le t'\) (thresholds may vary across agents). Then
\[
\E{}{v(S_1) + v(S_2) + 4v(T_1) + 4v(T_2)}
\;\ge\;
v(C)-2t',
\]
where the expectation is over the internal randomness of \textsc{NM-ReD} (the random subsampling into \(T_1,T_2\) and the final random choice of the output set).
\end{lemma}

We output $S_1$ and $S_2$ with probability $3/10$ each and $T_1$,$T_2$ with probability $2/10$ each. Therefore 
$$\E{}{v(T)} = \E{}{\frac{3}{10}v(S_1) + \frac{3}{10}v(S_2) + \frac{2}{10}v(T_1) + \frac{2}{10}v(T_2)}  \geq \frac{1}{20} \E{}{v(S_1) + v(S_2) + 4v(T_1) + 4v(T_2)} $$

Conditioning on \(\mathcal{E},\mathcal{F}\), and \(A\), the set \(C\) defined above satisfies \(v(C)\ge \sum_{a\in C}w(a)\ge 5t'\), and every agent in \(C\) was offered a price computed using a threshold \(t\leq t'\).
Applying Lemma~\ref{lem:amanat} yields
\[
\E{}{v(T)\mid \mathcal{E},\mathcal{F},A}
\;\ge\; \frac{3}{20}t' \geq 
\frac{3\OPT}{4000}.
\]

\textbf{Case 2: negation of \emph{A}.}
Conditioning on \(\mathcal{E}\) and \(\mathcal{F}\), Lemma~\ref{lem:dichotomy nm} implies that \textsc{NM-ReD} completes at least \(R/8\) successful descents. Let \(d\) denote the number of successful descents. Each successful descent contributes at least \(49\,\vmax\) total marginal value to \(S_1\) and \(S_2\), hence
\[
v(S_1)+v(S_2)
\;\ge\;
49\,\vmax\cdot d
\;\ge\;
49\,\vmax\cdot \frac{R}{8}
\;\ge\;
\frac{\OPT}{3920}.
\]

We next introduce the following powerful result due to \citet{Feige2011}.

\begin{theorem}[\cite{Feige2011}]\label{thm:Feige}
    Let $v: 2^A \to \mathbb{R}_{\geq 0}$ be a submodular function and let $S$ denote a random subset of $A$, where each element is included with probability $1/2$. Then 
    $$ \E{}{v(T)} \geq \frac{1}{4} \OPT(A),$$
    where $\OPT(A)$ denotes the value of the maximum value subset of $A$.
\end{theorem}

Applying Theorem~\ref{thm:Feige} to the random subsampling step  gives
\[
\E{}{v(T_1)} + \E{}{v(T_2)}
\;\ge\;
\frac{1}{4}\bigl(v(S_1)+v(S_2)\bigr).
\]
Since \textsc{NM-ReD} outputs \(S_1,S_2\) with probability \(3/10\) each and \(T_1,T_2\) with probability \(2/10\) each, we obtain
\[
\E{}{v(T)\mid \mathcal{E},\mathcal{F},\neg A}
\;\ge\;
\frac{\OPT}{11200}.
\]

\noindent Combining the two cases,
\[
\E{}{v(T)\mid \mathcal{E},\mathcal{F}}
\;\ge\;
\min\left\{
\E{}{v(T)\mid \mathcal{E},\mathcal{F},A},\;
\E{}{v(T)\mid \mathcal{E},\mathcal{F},\neg A}
\right\}
\;\ge\;
\frac{\OPT}{11200}.
\]
Finally, removing conditioning and using \(\Prob{\mathcal{E}}\ge \tfrac12\) and \(\Prob{\mathcal{F}\mid \mathcal{E}}\ge \tfrac12\),
\[
\E{}{v(T)}
\;\ge\;
\Prob{\mathcal{E}}\,\Prob{\mathcal{F}\mid \mathcal{E}}\cdot \E{}{v(T)\mid \mathcal{E},\mathcal{F}}
\;\ge\;
\tfrac12\cdot\tfrac12\cdot \frac{\OPT}{11200}
\;=\;
\frac{\OPT}{44800}.
\]
\end{proof}

\subsection{Small Market}

It remains to handle the case of small markets, i.e., instances in which
\( \OPT \;\le\; 50000\,\vmax.\) We use the classical Dynkin secretary algorithm on the entire sequence. Since valuations are submodular, this guarantees that the value obtained is at least the value of the selected agent. Dynkin’s algorithm selects the maximum singleton value with probability at least \(1/e\), and hence achieves expected value at least \(\vmax/e\).

In the small-market regime, \(\vmax \geq \frac{\OPT}{50000}\). Therefore, Dynkin’s algorithm is $136000$-competitive, and is truthful and budget feasible by construction.

Randomizing over Dynkin's algorithm and \textsc{NM-ReD} yields a $181000$-competitive posted-price mechanism for non-monotone submodular valuations in all regimes.

\section{Lower Bound for XOS Valuations}

In this section we present a logarithmic lower bound for online procurement auctions with XOS valuation functions. This construction is inspired by the lower bound framework of \citet{babaioff}. As defined in Section \ref{sec:prelim}, the adversary selects the cost of each agent and the valuation function. Then the agents arrive in a uniformly random order. The mechanism has query access to the set of revealed agents. Our main result is the following.

\begin{theorem}
\label{thm:xos-lower-bound}
For procurement auctions with XOS valuation functions, no online budget-feasible mechanism can achieve a competitive ratio better than
$\Omega\!\left(\tfrac{\log n}{(\log\log n)^2}\right)$.
\end{theorem}

\paragraph{Lower Bound Instance.}
We first describe a distribution over valuation functions; the actual valuation is drawn from this distribution. Fix a partition of the agent set $N$ into $k$ disjoint groups $S_1,\dots,S_k$, each of size $m:=\tfrac{\log n}{\log\log n}$, where $k:=\lfloor n/m\rfloor$. For each agent $a\in N$, sample an independent Bernoulli random variable $\Active(a)\in\{0,1\}$ with $\Prob{\Active(a)=1}=p:=\tfrac{\log\log n}{\log n}$. Let $\mathcal{D}$ denote the resulting product distribution, and let $I$ be a realization.

For any set $T\subseteq N$ and realization $I$, define the valuation
\[
v_I(T)=\max_{i\in[k]} \mu_i(T),
\qquad
\mu_i(T):=\sum_{a\in S_i\cap T}\Active(a,I).
\]
Each clause $\mu_i$ is additive and supported on a single group; hence $v_I$ is XOS. 

Finally, we set the cost for all agents $a \in N$ to be $c(a)=pB$. Thus, the budget allows the purchase of at most $1/p=m$ agents. 

\paragraph{Lower Bound Analysis.} For each group $i\in[k]$, let $X_i:=\sum_{a\in S_i}\Active(a,I)$ denote the number of active agents in group $i$. Then $X_i\sim\mathrm{Bin}(m,p)$, the variables $(X_i)_{i\in[k]}$ are independent, and $\mathbb{E}[X_i]=mp=1$. Intuitively, each group contains only a constant number of active agents in expectation, but since there are many groups, the group with the most active agents is significantly larger. We continue by formalizing and proving the latter statement in the following lemma.

\begin{lemma}
\label{lem:opt-xos}
For the above instance distribution,
\[
\mathbb{E}_{I\sim\mathcal{D}}[\OPT(I)]
=\Theta\!\left(\frac{\log n}{\log\log n}\right).
\]
\end{lemma}

\begin{proofsketch}
We first characterize the offline optimum.
Since $c(a)=pB$, any feasible solution can select at most $m$ agents.
For any set $T\subseteq N$,
\[
v_I(T)=\max_{i\in[k]}\sum_{a\in S_i\cap T}\Active(a,I)
\le \max_{i\in[k]} X_i.
\]
Conversely, for any group $i$, the offline optimum can select all active agents in $S_i$ (since $X_i\le m$) and obtain value $X_i$.
Hence $\OPT(I)=\max_{i\in[k]} X_i$.

Let $M:=\max_{i\in[k]} X_i$.
A standard extreme-value analysis shows that
$\mathbb{E}[M]=\Theta(\tfrac{\log k}{\log\log k})$. Substituting $k=\Theta(n\log\log n/\log n)$ yields the claimed bound. The full calculation is standard and deferred to Appendix~\ref{proof of xos}. 
\end{proofsketch}

We now proceed to upper bound the expected performance of any online mechanism on this instance.

\begin{lemma}
\label{lem:alg-upper}
For the hard distribution $\mathcal{D}$, every budget-feasible online mechanism satisfies
\[
\mathbb{E}_{I\sim\mathcal{D}}[\ALG(I)] = O(\log\log n).
\]
\end{lemma}

\begin{proof}
Recall that each agent has cost $c(a)=pB$ with $p=\tfrac{\log\log n}{\log n}$, hence any budget-feasible mechanism can hire at most
$m:=1/p=\tfrac{\log n}{\log\log n}$ agents in total.

Fix a run of the mechanism (arrival order and internal randomness).
For each group $i\in[k]$, define $\tau_i$ to be the first time (arrival index) at which the mechanism hires an agent from group $S_i$; if the mechanism never hires from $S_i$, set $\tau_i=\infty$.
Let $G:=\{i\in[k]: \tau_i<\infty\}$ be the set of groups from which the mechanism hires at least one agent.
Since the mechanism hires at most $m$ agents overall, we have $|G|\le m$.

For each $i\in G$, define the random variable
\[
R_i := \sum_{a\in S_i : \text{$a$ arrives after time }\tau_i} \Active(a,I),
\]
i.e., the number of \emph{remaining active agents} in group $i$ after the mechanism first commits to (hires from) that group.

Let $T$ be the set of hired agents. Since the valuation is $v_I(T)=\max_{i\in[k]} \sum_{a\in S_i\cap T}\Active(a,I)$, the mechanism's value is the number of active hired agents in the best group it hires from.
Fix any group $i\in G$ and consider the mechanism's contribution from that group.
At time $\tau_i$ it hires exactly one agent from $S_i$ for the first time; this single agent contributes at most $1$ to the realized value.
Any additional active hired agents from $S_i$ must arrive \emph{after} time $\tau_i$, and their total number is at most $R_i$.
Therefore, for every realization we have the pointwise bound
\(
\ALG(I) \le 1 + \max_{i\in G} R_i.
\)

Conditional on $\tau_i$ and on the entire history up to time $\tau_i$, the activation bits of agents in $S_i$ that arrive after time $\tau_i$ remain independent $\mathrm{Bernoulli}(p)$, because activations are sampled independently across agents at the start and are not affected by the mechanism's actions.
Moreover, the number of agents of $S_i$ arriving after time $\tau_i$ is at most $m-1$.
Hence, for each $i\in G$, the random variable $R_i$ is stochastically dominated by
$\mathrm{Bin}(m,p)$ (or more tightly by $\mathrm{Bin}(m-1,p)$), regardless of how $\tau_i$ is chosen.

Let $L_m:=2\left\lceil \frac{\log m}{\log\log m}\right\rceil$. Since $|G| \leq m$, it follows that

$$\E{}{\alg(I)} \leq 1 + \sum_{q = 1}^{m} \Prob{\max_{i \in G} R_i \geq  q} \leq 1 + L_m + \sum_{q = L_m}^{m} \Prob{\max_{i \in G} R_i \geq q} $$

By union bound we get:
\begin{equation}
    \E{}{\alg} \leq 1 + L_m + \sum_{q = L_m}^{m} m \cdot \Prob{ R_i \geq q}
\end{equation}
Applying Chernoff's bound gives:
\begin{equation}
        \E{}{\alg} \leq 1 + L_m + \sum_{q = L_m}^{m} m \cdot \lp( \frac{e}{q} \rp)^{q}
\end{equation}
Finally, we need to notice that \begin{equation}
    \lim_{m\to +\infty }\sum_{q = L_m}^{m} m \cdot \lp( \frac{e}{q} \rp)^{q} =0
\end{equation}
Indeed, for $q\ge L_m$, the terms $\left(\frac{e}{q}\right)^q$ decrease geometrically for large $m$, so
\[
\sum_{q=L_m}^{m} m\left(\frac{e}{q}\right)^q
\le
2m\left(\frac{e}{L_m}\right)^{L_m}.
\]
Moreover,
\[
\log\left(m\left(\frac{e}{L_m}\right)^{L_m}\right)
=
\log m-L_m(\log L_m-1)
\le -\Omega(\log m),
\]
because $L_m=\Theta(\log m/\log\log m)$ and $\log L_m=\Theta(\log\log m)$. Hence
\[
\sum_{q=L_m}^{m} m\left(\frac{e}{q}\right)^q=o(1).
\]
Which gives the desired result:
\begin{align*}
        \E{}{\alg} \leq 1 + 2\cdot \left\lceil \frac{\log m}{\log\log m}\right\rceil + o(1) 
\end{align*}

\end{proof}

Combining Lemma~\ref{lem:opt-xos} together with Lemma~\ref{lem:alg-upper} completes the proof of Theorem~\ref{thm:xos-lower-bound}.

\begin{remark}
An alternative model for online access to the valuation function can be derived by \cite[Definition~2.1]{rubi-definition}. Our lower bound extends to that model by replacing each agent $i$ of the original instance with a pair $(U_i,D_i)$. We select \(U_i=\{a_i,b_i\}\), such that only \(a_i\) contributes to the corresponding additive clause, and select \(D_i=a_i\) with probability \(\log\log n/\log n\) (and \(b_i\) otherwise). 
\end{remark}

\section{Conclusion}

We studied online budget-feasible procurement auctions under random-order arrivals, focusing on the power of sequential posted-price mechanisms. Our main technical contribution is the \emph{Repeated Descent} (\RED) subroutine, a deterministic pricing primitive that enforces budget feasibility while dynamically steering the mechanism toward regimes where linear pricing is effective, without ever estimating the optimal pricing scale.

Using \RED, we significantly improve the state of the art for posted-price mechanisms with monotone submodular valuations, reducing the competitive ratio by several orders of magnitude compared to previous work. We further show that constant-competitive guarantees extend beyond monotone objectives: by combining \RED with random subsampling techniques, we obtain the first constant-competitive posted-price mechanism for online procurement with non-monotone submodular valuations. Our positive results are complemented by a logarithmic lower bound for XOS valuations, which applies to all online budget-feasible mechanisms. This impossibility result delineates a sharp boundary between valuation classes that admit constant-competitive online mechanisms and those for which such guarantees are unattainable, even under random-order arrivals.

Finally, an interesting question raised by our work concerns the gap between submodular and XOS valuations in the online setting. While our lower bound rules out constant-competitive guarantees for XOS valuations, it remains open whether one can design an online mechanism with a polylogarithmic competitive ratio, even in the more powerful direct-revelation model. 



\bibliographystyle{ACM-Reference-Format}
\bibliography{sample-bibliography}

@String{Computing = "Computing" }

@String{Computer = "{IEEE} Computer" }

@String{Springer = "Springer-Verlag" }

@article{Sviri04,
  author       = {Maxim Sviridenko},
  title        = {A note on maximizing a submodular set function subject to a knapsack constraint},
  journal      = {Operations Research Letters},
  volume       = {32},
  number       = {1},
  pages        = {41--43},
  year         = {2004},
}

@article{NemhauserWF78,
  author       = {George L. Nemhauser and
                  Laurence A. Wolsey and
                  Marshall L. Fisher},
  title        = {An analysis of approximations for maximizing submodular set functions
                  - {I}},
  journal      = {Mathematical Programing},
  volume       = {14},
  number       = {1},
  pages        = {265--294},
  year         = {1978},
 }

@article{BadaDO19,
  author       = {Ashwinkumar Badanidiyuru and
                  Shahar Dobzinski and
                  Sigal Oren},
  title        = {Optimization with Demand Oracles},
  journal      = {Algorithmica},
  volume       = {81},
  number       = {6},
  pages        = {2244--2269},
  year         = {2019},
 }

@inproceedings{Bada2012,
  author       = {Ashwinkumar Badanidiyuru and                   Robert Kleinberg and Yaron Singer},
  title        = {Learning on a budget: posted price mechanisms for online procurement},
  booktitle    = {Proc.  of the 13th {ACM} Conference on Electronic Commerce ({EC}~2012)},
  pages        = {128--145},
  publisher    = {{ACM}},
  year         = {2012},
}

@article{KhullerMN99,
  author       = {Samir Khuller and Anna Moss and Joseph Naor},
  title        = {The Budgeted Maximum Coverage Problem},
  journal      = {Information Processing Letters},
  volume       = {70},
  number       = {1},
  pages        = {39--45},
  year         = {1999},
}

@inproceedings{Singer10,
  author       = {Yaron Singer},
  title        = {Budget Feasible Mechanisms},
  booktitle    = {Proc. of the 51th {IEEE} Symposium on Foundations of Computer Science ({FOCS}~2010)}, 
  pages        = {765--774},
  publisher    = {{IEEE} Computer Society},
  year         = {2010},
}

@article{Singer13,
  author       = {Yaron Singer},
  title        = {Budget feasible mechanism design},
  journal      = {SIGecom Exchanges},
  volume       = {12},
  number       = {2},
  pages        = {24--31},
  year         = {2013},
}

@inproceedings{ChenGL2011,
  author       = {Ning Chen and
                  Nick Gravin and
                  Pinyan Lu},
  title        = {On the Approximability of Budget Feasible Mechanisms},
  booktitle    = {Proc. of the 22nd {ACM-SIAM} Symposium on Discrete Algorithms ({SODA}~2011)},
  pages        = {685--699},
  publisher    = {{SIAM}},
  year         = {2011},
}

@inproceedings{BeiCGL2012,
  author       = {Xiaohui Bei and
                  Ning Chen and
                  Nick Gravin and
                  Pinyan Lu},
  title        = {Budget feasible mechanism design: from prior-free to bayesian},
  booktitle    = {Proc. of the 44th Symposium on Theory of Computing Conference ({STOC}~2012)},
  pages        = {449--458},
  publisher    = {{ACM}},
  year         = {2012},
}

@article{GravinJLZ2020,
  author       = {Nick Gravin and
                  Yaonan Jin and
                  Pinyan Lu and
                  Chenhao Zhang},
  title        = {Optimal Budget-Feasible Mechanisms for Additive Valuations},
  journal      = {{ACM} Transactions on Economics and Computation},
  volume       = {8},
  number       = {4},
  pages        = {21:1--21:15},
  year         = {2020},
}

@inproceedings{AnariGN2014,
  author       = {Nima Anari and
                  Gagan Goel and
                  Afshin Nikzad},
  title        = {Mechanism Design for Crowdsourcing: An Optimal $1-1/e$ Competitive Budget-Feasible
                  Mechanism for Large Markets},
  booktitle    = {Proc. of the 55th {IEEE} Symposium on Foundations of Computer Science ({FOCS}~2014)},
  pages        = {266--275},
  publisher    = {{IEEE} Computer Society},
  year         = {2014},
}

@inproceedings{LCLW2024_survey,
  author       = {Xiang Liu and
                  Hau Chan and
                  Minming Li and
                  Weiwei Wu},
  title        = {Budget Feasible Mechanisms: {A} Survey},
  booktitle    = {Proc. of the 33rd International Joint Conference on Artificial Intelligence ({IJCAI}~2024)},
  pages        = {8132--8141},
  publisher    = {ijcai.org},
  year         = {2024},
}

@article{JalalyT2021,
  author       = {Pooya Jalaly and
                  {\'{E}}va Tardos},
  title        = {Simple and Efficient Budget Feasible Mechanisms for Monotone Submodular Valuations},
  journal      = {{ACM} Transactions on Economics and Computation},
  volume       = {9},
  number       = {1},
  pages        = {4:1--4:20},
  year         = {2021},
 }

@inproceedings{BalkanskiGGST2022,
  author       = {Eric Balkanski and
                  Pranav Garimidi and
                  Vasilis Gkatzelis and
                  Daniel Schoepflin and
                  Xizhi Tan},
  title        = {Deterministic Budget-Feasible Clock Auctions},
  booktitle    = {Proce. of the 2022 {ACM-SIAM} Symposium on Discrete Algorithms ({SODA}~2022)},
  pages        = {2940--2963},
  publisher    = {{SIAM}},
  year         = {2022},
}

@inproceedings{HanWHC23,
  author       = {Kai Han and
                  You Wu and
                  He Huang and
                  Shuang Cui},
  title        = {Triple Eagle: Simple, Fast and Practical Budget-Feasible Mechanisms},
  booktitle    = {Proc. of the 37th Conference on Neural Information Processing Systems ({NeurIPS}~2023)},
  year         = {2023},
  numpages     = {18},
}

@article{AmanatidisKS22,
  author       = {Georgios Amanatidis and
                  Pieter Kleer and
                  Guido Sch{\"{a}}fer},
  title        = {Budget-Feasible Mechanism Design for Non-monotone Submodular Objectives:
                  Offline and Online},
  journal      = {Mathematics of Operations Research},
  volume       = {47},
  number       = {3},
  pages        = {2286--2309},
  year         = {2022},
}

@inproceedings{AmanatidisBM17,
  author       = {Georgios Amanatidis and
                  Georgios Birmpas and
                  Evangelos Markakis},
  title        = {On Budget-Feasible Mechanism Design for Symmetric Submodular Objectives},
  booktitle    = {Proc. of the 13th Conference on Web and Internet Economics ({WINE}~2017)},
  series       = {Lecture Notes in Computer Science},
  volume       = {10660},
  pages        = {1--15},
  publisher    = {Springer},
  year         = {2017},
}

@inproceedings{DobzinskiPS11,
  author       = {Shahar Dobzinski and
                  Christos H. Papadimitriou and
                  Yaron Singer},
  title        = {Mechanisms for complement-free procurement},
  booktitle    = {Proc. of the 12th {ACM} Conference on Electronic Commerce ({EC}~2011)},
  pages        = {273--282},
  publisher    = {{ACM}},
  year         = {2011},
}

@inproceedings{HuangHC023,
  author       = {He Huang and
                  Kai Han and
                  Shuang Cui and
                  Jing Tang},
  title        = {Randomized Pricing with Deferred Acceptance for Revenue Maximization
                  with Submodular Objectives},
  booktitle    = {Proc. of the 32nd {ACM} Web Conference 2023, ({WWW}~2023)},
  pages        = {3530--3540},
  publisher    = {{ACM}},
  year         = {2023},
}

@inproceedings{swamy-subadditive,
  title={An O (log log n)-approximate budget feasible mechanism for subadditive valuations},
  author={Neogi, Rian and Pashkovich, Kanstantsin and Swamy, Chaitanya},
  booktitle={Proceedings of the 26th ACM Conference on Economics and Computation},
  pages={599--599},
  year={2025}
}

@inproceedings{FeldmanNS11,
  author       = {Moran Feldman and
                  Joseph Naor and
                  Roy Schwartz},
  title        = {Improved Competitive Ratios for Submodular Secretary Problems (Extended
                  Abstract)},
  booktitle    = {Proc. of the 14th Conference on Approximation Algorithms for Combinatorial Optimization Problems ({APPROX}~2011) and of the 15th Conference on on Randomization and Computation ({RANDOM}~2011)}, 
  series       = {Lecture Notes in Computer Science},
  volume       = {6845},
  pages        = {218--229},
  publisher    = {Springer},
  year         = {2011},
 }

@article{BateniHZ13,
  author       = {MohammadHossein Bateni and
                  Mohammad Taghi Hajiaghayi and
                  Morteza Zadimoghaddam},
  title        = {Submodular secretary problem and extensions},
  journal      = {{ACM} Transactions on Algorithms},
  volume       = {9},
  number       = {4},
  pages        = {32:1--32:23},
  year         = {2013},
}

@inproceedings{KesselheimT17,
  author       = {Thomas Kesselheim and
                  Andreas T{\"{o}}nnis},
  title        = {Submodular Secretary Problems: Cardinality, Matching, and Linear Constraints},
  booktitle    = {Proc. of the 20th Conference on Approximation Algorithms for Combinatorial Optimization Problems ({APPROX}~2017) and of the 21st Conference on on Randomization and Computation ({RANDOM}~2017)}, 
  series       = {LIPIcs},
  volume       = {81},
  pages        = {16:1--16:22},
  publisher    = {Schloss Dagstuhl - Leibniz-Zentrum f{\"{u}}r Informatik},
  year         = {2017},
 }

@inproceedings{SingerM13,
  author       = {Yaron Singer and
                  Manas Mittal},
  title        = {Pricing mechanisms for crowdsourcing markets},
  booktitle    = {Proc. of the 22nd International World Wide Web Conference ({WWW}~2013)}, 
  pages        = {1157--1166},
  publisher    = {{ACM}},
  year         = {2013},
}

@article{chen2017bernstein,
  title={A Bernstein type inequality for NOD random variables and applications},
  author={Chen, PY and Sung, Soo Hak},
  journal={J. Math. Inequal},
  volume={11},
  pages={455--467},
  year={2017}
}

@misc{negass,
    author = {David Wajc},
    title = {Negative association: definition, properties, and applications} ,
    year = {2017},
    link = {https://www.cs.cmu.edu/~dwajc/notes/Negative%20Association.pdf}
}

@inbook{CFPT2025,
author = {Charalampopoulos, Andreas and Fotakis, Dimitris and Patsilinakos, Panagiotis and Tolias, Thanos},
title = {A Competitive Posted-Price Mechanism for Online Budget-Feasible Auctions},
year = {2025},
isbn = {9798400719431},
publisher = {Association for Computing Machinery},
address = {New York, NY, USA},
url = {https://doi.org/10.1145/3736252.3742668},
booktitle = {Proceedings of the 26th ACM Conference on Economics and Computation},
pages = {1046–1075},
numpages = {30}
}

@inproceedings{babaioff,
author = {Babaioff, Moshe and Immorlica, Nicole and Kleinberg, Robert},
title = {Matroids, secretary problems, and online mechanisms},
year = {2007},
isbn = {9780898716245},
publisher = {Society for Industrial and Applied Mathematics},
address = {USA},
booktitle = {Proceedings of the Eighteenth Annual ACM-SIAM Symposium on Discrete Algorithms},
pages = {434–443},
numpages = {10},
location = {New Orleans, Louisiana},
series = {SODA '07}
}

@article{Feige2011,
author = {Feige, Uriel and Mirrokni, Vahab S. and Vondr\'{a}k, Jan},
title = {Maximizing Non-monotone Submodular Functions},
journal = {SIAM Journal on Computing},
volume = {40},
number = {4},
pages = {1133-1153},
year = {2011},
doi = {10.1137/090779346},

URL = { 
    
        https://doi.org/10.1137/090779346
    
    

},
eprint = { 
    
        https://doi.org/10.1137/090779346
    
    

}
}

@inproceedings{rubi-definition,
author = {Rubinstein, Aviad and Singla, Sahil},
title = {Combinatorial prophet inequalities},
year = {2017},
publisher = {Society for Industrial and Applied Mathematics},
address = {USA},
pages = {1671–1687},
numpages = {17},
location = {Barcelona, Spain},
series = {SODA '17}
}

\newpage
\appendix

\section{Application of Bernstein's inequality: Proof of Lemma~\ref{lemma:Bern-appl}}\label{bern}

Fix a round $j$ and recall the random variables $(X_i^j)_{i\in[\ell]}$ defined in \ref{sub:rounds-conc}: for each position $i$ of round $j$,
\[
X_i^j =
\begin{cases}
w(a) & \text{if agent } a \in C^*\setminus\{a_{\max}\} \text{ occupies position } i,\\
0 & \text{otherwise}.
\end{cases}
\]
Let $ X^j := \sum_{i=1}^{\ell} X_i^j$ and define $g(C^*) := \sum_{a\in C^*\setminus\{a_{\max}\}} w(a).$\\

Conditioned on event $\mathcal{E}$ (that $a_{\max}$ lies in the first half), the second half of the sequence is a uniformly random permutation of a uniformly random subset of size $n/2$ out of the $n-1$ agents in $N\setminus\{a_{\max}\}$. Hence, for every position $i$ in the second half,
\[
\mu_i := \mathbb{E}[X_i^j \mid \mathcal{E}] = \frac{g(C^*)}{n-1}.
\]
Moreover, since $0\le X_i^j \le v_{\max}$, we have
\[
X_i^j - \mu_i \le v_{\max} - \mu_i
\quad\Rightarrow\quad
b := v_{\max}-\frac{g(C^*)}{n-1}.
\]
The expectation of $X^j$ is
\[
\mu := \mathbb{E}[X^j \mid \mathcal{E}]
= \sum_{i=1}^{\ell} \mu_i
= \frac{\ell}{n-1}\,g(C^*).
\]
Using $0\le X_i^j \le v_{\max}$, we have $(X_i^j)^2 \le v_{\max}X_i^j$, and thus
\[
\mathbb{E}[(X_i^j)^2 \mid \mathcal{E}] \le v_{\max} \cdot \mu_i.
\]
Therefore,
\[
\mathrm{Var}(X_i^j \mid \mathcal{E})
= \mathbb{E}[(X_i^j-\mu_i)^2\mid \mathcal{E}]
\le v_{\max}\mu_i - \mu_i^2
= \frac{g(C^*)}{n-1}\left(v_{\max}-\frac{g(C^*)}{n-1}\right).
\]
Summing over $i\in[\ell]$ yields
\[
V := \sum_{i=1}^{\ell} \mathrm{Var}(X_i^j \mid \mathcal{E})
\le \frac{\ell}{n-1} g(C^*)\left(v_{\max}-\frac{g(C^*)}{n-1}\right).
\]

\begin{lemma}\label{lemma:Bern}
For every round $j$,
\[
\Prob{X^j \le 160\,v_{\max}\,\middle|\,\mathcal{E}} \le 0.1.
\]
\end{lemma}

\begin{proof}
We apply Bernstein's inequality under negative association. Set $t := \mu/6$, so that $\mu-t = 5\mu/6$. It suffices to show $5\mu/6 \ge 160\,v_{\max}$ and then bound
$\Prob{X^j \le \mu-t \mid \mathcal{E}}$.

First,
\[
\frac{5\mu}{6}
= \frac{5}{6}\cdot \frac{\ell}{n-1}\,g(C^*)
= \frac{5}{6}\cdot \frac{50n}{n-1}\cdot \frac{v_{\max}}{t'}\,g(C^*)
\ge 160\,v_{\max},
\]
where the last inequality uses the large-market consequences $\ell = 50\frac{v_{\max}}{t'}n \ge 400\frac{v_{\max}}{\OPT}n$ and $g(C^*)\ge \frac{44}{90}\OPT$. Bernstein's inequality gives
\[
\Prob{X^j \le \mu - t \,\middle|\, \mathcal{E}}
\le \exp\!\left(
-\frac{t^2}{2V+\frac{2bt}{3}}
\right)
=
\exp\!\left(
-\frac{\mu^2/36}{2V+\frac{b\mu}{9}}
\right).
\]
Substituting the bounds on $V$ and $b$ and simplifying yields
\[
\Prob{X^j \le \mu - t \,\middle|\, \mathcal{E}}
\le
\exp\!\left(
-\frac{\frac{\ell}{n-1}g(C^*)}
{76\left(v_{\max}-\frac{g(C^*)}{n-1}\right)}
\right)
\le
\exp\!\left(
-\frac{\ell\,g(C^*)}{76\,n\,v_{\max}}
\right),
\]
where we used $v_{\max}-\frac{g(C^*)}{n-1}\le v_{\max}$ and $n-1\ge n/2$ for large $n$.\\

Finally, using $\ell \ge 400\frac{v_{\max}}{\OPT}n$ and $g(C^*)\ge \frac{44}{90}\OPT$ gives an exponent at most $-2.31$, hence
\[
\Prob{X^j \le 160\,v_{\max}\,\middle|\,\mathcal{E}}
\le e^{-2.31} \le 0.1.
\]
\end{proof}

\section{Proof of Theorem~\ref{thm:mediummarket}}\label{sec:mediummarket}

\textsc{MediumMarket} uses an initial part of the input sequence to obtain a rough estimate of $\OPT$ (it satisfies $\vmax \leq \OPT \leq n\cdot \vmax$). Then, it randomly selects a multiplier from a carefully chosen set of multipliers and uses it to post linear prices.

\begin{algorithm}[ht]
\caption{\textsc{MediumMarket}}
\KwInput{Set of agents $N$, budget $B$}

Sample $\tau \sim \mathrm{Bin}(|N|,1/2)$, $T \gets \varnothing$\;
Offer price $0$ to the first $\tau$ arriving agents and observe $\widehat\vmax$\;
Choose a multiplier $m$ uniformly at random from
$\mathcal T=\{5.89,\,18.76,\,59.00,\,184.85,\,578.62,\,1809.01\}$\;
Set threshold $\hat t \gets m\cdot \widehat \vmax$\;

\ForEach{remaining agent $a$ in arrival order}{
    \If{$\frac{v(a\mid T)\cdot B}{\hat t} \le B\; $}{
    Offer price $p \gets \frac{v(a\mid T)\cdot B}{\hat t}$ to agent $a$\;
    \If{$p\ge c(a)$}{
        $T \gets T \cup \{a\}$, $B \gets B - p$\;
    }}
    $N \gets N \setminus \{a\}$
}
\KwRet{$T$}
\end{algorithm}

\begin{theorem}\label{thm:mediummarket}
If $29 \cdot v_{\max} \le \OPT \le 8000 \cdot v_{\max}$,
then \textsc{MediumMarket} is $326$-competitive.
\end{theorem}

The proof follows the framework of ~\cite{Bada2012}, adapted to our randomized-threshold implementation and constant-size multiplier net. We include it for completeness.

Let $S^*=\{a_1,a_2,\ldots,a_k\}$ be an optimal solution, ordered so that agents
appear in non-increasing order of marginal contribution. Define
\[
w(a_i)
:= v(\{a_1,\ldots,a_i\}) - v(\{a_1,\ldots,a_{i-1}\}),
\]
so that $\sum_{i=1}^k w(a_i)=\OPT$ and $0\le w(a_i)\le v_{\max}$ for all $i$.\\

The mechanism splits the agents into an exploration phase and an exploitation phase, with each agent independently assigned to each phase with probability $1/2$.

Let $\mathcal E_{\max}$ be the event that at least one agent attaining value $v_{\max}$
appears in the exploration phase. Since the split is uniform,
$\Prob{\mathcal E_{\max}}=1/2$, and on this event the mechanism learns the correct scale $v_{\max}$.

For each $a_i\in S^*$ define the random variable
\[
X_i :=
\begin{cases}
w(a_i) & \text{if } a_i \text{ belongs to the exploitation phase},\\
0 & \text{otherwise}.
\end{cases}
\]
The variables $(X_i)_i$ are independent. By submodularity, the value of the set of agents that belong in the optimal solution and appear in the exploitation phase, denoted $S_2^*$, satisfies
\[
v(S_2^*) \;\ge\; \sum_{i=1}^k X_i.
\]

\begin{lemma}\label{lem:mm-split}
Assume $\OPT \ge 29\,v_{\max}$. Then
\[
\Prob{\sum_{i=1}^k X_i \le \frac{\OPT}{4}} \;\le\; 0.07.
\]
\end{lemma}

\begin{proof}
Each $X_i$ takes values in $[0,v_{\max}]$ and satisfies $\E{}{X_i}=\tfrac12 w(a_i)$.
Let $\mu=\E{}{\sum_i X_i}=\OPT/2$.
Applying Bernstein’s inequality (Appendix~\ref{bern}) with $b=v_{\max}$ and
\[
V
=\sum_i \E{}{(X_i-\E{}{X_i})^2}
\le \tfrac14\sum_i w(a_i)^2
\le \tfrac14\,v_{\max}\OPT
\le \tfrac{\OPT^2}{4\cdot 29},
\]
yields the stated bound after substitution.
\end{proof}

By symmetry, the same bound holds for the exploration phase.
Let $\mathcal E_{\text{split}}$ denote the event that both phases contain optimal value at least $\OPT/4$.
By a union bound,
\[
\Prob{\mathcal E_{\text{split}}} \ge 1-2\cdot 0.07 = 0.86.
\]
Combining with $\mathcal E_{\max}$ and union bounding bad events gives
\[
\Prob{\mathcal E_{\max}\cap \mathcal E_{\text{split}}} \ge 1-0.5-0.14 = 0.36.
\]

Condition on $\mathcal E_{\max}\cap \mathcal E_{\text{split}}$. Then $v(S_2^*) \ge \OPT/4$
and the mechanism uses the correct scale $v_{\max}$. Suppose
\[
v(S_2^*) \in [\ell\,v_{\max},\,u\,v_{\max})
\]
for some $\ell,u>0$, and the mechanism posts linear prices with threshold
$t = m\cdot v_{\max}$.

\begin{lemma}\label{lem:mm-fixed-m}
Under the above conditions, the mechanism collects value at least
\[
\min\!\left\{\frac{\ell-m}{\ell},\ \frac{m-1}{u}\right\}\cdot v(S_2^*).
\]
\end{lemma}

\begin{proof}
If all the agents in $S_2^*$ accept, the mechanism collects $v(S_2^*)$. Otherwise, a buyer in $S_2^*$ is excluded from the solution. This can occur due to either "budget exhaustion" or "offer rejection". We analyze the two conditions separately.  

\begin{enumerate}
    \item \emph{Budget exhaustion.} The remaining budget is less than the price of the next agent, which is at most $v_{\max}B/t = B/m$. Hence at least $(m-1)B/m$ budget has been spent. Since all purchases were made at threshold $t$, the collected value is at least $(m-1)v_{\max} \ge \frac{m-1}{u}v(S_2^*)$.

    \item \emph{Rejections.} If some agents in $S_2^*$ reject, then for each rejected agent $a$, $c(a) > v(a\mid T)\cdot B/t$. Summing over rejected agents and using submodularity yields
    \[ v(S_2^*) - v(T) \le m\,v_{\max}, \]
    and therefore
    \[ v(T) \ge (\ell-m)v_{\max} = \frac{\ell-m}{\ell}v(S_2^*). \]
\end{enumerate}
\end{proof}

\paragraph{Multiplier net for the range $[29/4,\,8000]$.} Let $\ell_1=29/4$ and let
\(
r := \left(\frac{8000}{\ell_1}\right)^{1/6}.
\) Define intervals
\[
[\ell_i,u_i) \quad\text{with}\quad
\ell_i=\ell_1 r^{i-1},\; u_i=r\ell_i,
\qquad i=1,\dots,6.
\]
These intervals cover the range $v(S_2^*)/v_{\max} \in [29/4,\,8000]$.

For each interval define
\[
m_i := \frac{\ell_i(u_i+1)}{\ell_i+u_i},
\]
which equalizes the two terms in Lemma~\ref{lem:mm-fixed-m}.
Let $\mathcal T=\{m_1,\dots,m_6\}$; numerically,
\[
\mathcal T=\{5.89,\,18.76,\,59.00,\,184.85,\,578.62,\,1809.01\}.
\]

Fix any interval $[\ell_i,u_i)$. Plugging $m=m_i$ into Lemma~\ref{lem:mm-fixed-m} gives
\[
\min\!\left\{\frac{\ell_i-m_i}{\ell_i},\ \frac{m_i-1}{u_i}\right\}
= \frac{\ell_i-1}{\ell_i+u_i}
= \frac{1-\frac1{\ell_i}}{1+r}
\;\ge\; \frac{1-\frac1{\ell_1}}{1+r}.
\]
Thus, conditioned on selecting the correct multiplier, the mechanism recovers at least
\[
\gamma \cdot v(S_2^*)
\quad\text{where}\quad
\gamma := \frac{1-\frac1{\ell_1}}{1+r}.
\]

\paragraph{Putting the pieces together.}
On the event $\mathcal E_{\max}\cap\mathcal E_{\text{split}}$ (probability at least $0.36$),
we have $v(S_2^*)\ge \OPT/4$.
Moreover, the correct multiplier is chosen with probability $1/6$.
Therefore,
\[
\E{}{\ALG}
\;\ge\;
0.36 \cdot \frac{1}{6} \cdot \gamma \cdot \frac{\OPT}{4}.
\]
With the above choice of $r$, the resulting constant satisfies
\[
0.36\cdot \frac{1}{6}\cdot \gamma \cdot \frac{1}{4} \;\ge\; \frac{1}{326},
\]
and hence $\E{}{\ALG} \ge \OPT/326$, proving that \textsc{MediumMarket} is $326$-competitive
for $29\,v_{\max}\le \OPT \le 8000\,v_{\max}$.

\section{Proof of Lemma \ref{lem:amanat}}\label{sec: amanat}

Recall that \textsc{NM-ReD} is always budget feasible by design, therefore every agent $a \notin (S_1 \cup S_2)$ has rejected an offer by the mechanism, i.e. $c(a) \geq p_a$. 

Recall that $C$ is a budget feasible set of agents such that each $a \in C$ was offered a price computed with some threshold $t \leq t'$. We partition $C$ into three sets: $C_1 = C \cap S_1$, $C_2 = C \cap S_2$ and $C_3 = C \setminus (C_1 \cup C_2)$. By subadditivity, it holds
$$ v(C) \leq v(C_1) + v(C_2) + v(C_3).$$

To upper bound the value of $C_1$ and $C_2$ we use Theorem \ref{thm:Feige}. This implies that for $j \in \{1,2\}$
$$ v(C_j) \leq 4 \E{}{v(T_j)}.$$

We proceed to upper bound the value of $v(C_3)$. By non negativity and submodularity, it holds
$$ v(C_3) \leq v(C_3) + v(C_3 \cup S_1 \cup S_2) \leq v(C_3 \cup S_1) + v(C_3 \cup S_2).$$

We continue by bounding $v(C_3 \cup S_1)$, by submodularity it holds
$$ v(C_3 \cup S_1) \leq v(S_1) + \sum_{a \in C_3 \setminus S_1} v(a \mid S_1).$$

We next use that each agent $a \in C_3 \setminus S_1$ rejected an offer with threshold at most $t'$ and marginal contribution at least $v(a \mid S_1)$. This leads to 
$$ \sum_{a \in C_3 \setminus S_1} v(a \mid S_1) \leq \frac{t'}{B} \sum_{a \in C_3 \setminus S_1} c(a) \leq t',$$
where the last inequality holds due to the budget feasibility of $C$. 

Therefore we get $v(C_3 \cup S_1) \leq v(S_1) + t'$ and with identical proof $v(C_3 \cup S_2) \leq v(S_2) + t'.$

Putting everything together concludes the proof of the Lemma.

\section{Proof of Lemma \ref{lem:opt-xos}}\label{proof of xos}

Fix $m:=\tfrac{\log n}{\log\log n}$ and $p:=\tfrac{\log\log n}{\log n}$, and set $k:=\lfloor n/m\rfloor$. For a realization $I\sim\mathcal{D}$, define $X_i:=\sum_{a\in S_i}active(a,I)$, so that $X_i\sim \mathrm{Bin}(m,p)$, the variables $(X_i)_{i\in[k]}$ are independent, and $\E{}{X_i}=mp=1$.

Each agent has cost $c(a)=pB$, hence the budget allows purchasing at most $B/c(a)=1/p=m$ agents. For any set $T\subseteq N$ and realization $I$, the valuation is $v_I(T)=\max_{i\in[k]} \mu_i(T)$ where $\mu_i(T)=\sum_{a\in S_i\cap T} \Active(a,I)$. Therefore, for any feasible $T$ with $c(T)\le B$, \[ v_I(T)=\max_{i\in[k]} \sum_{a\in S_i\cap T}\Active(a,I) \le \max_{i\in[k]} \sum_{a\in S_i}\Active(a,I) = \max_{i\in[k]} X_i. \] On the other hand, for every $i$ the offline optimum can pick exactly the active agents inside $S_i$ (or all of them, since $X_i\le m$) and obtain value $X_i$ while staying within budget. Hence \[ \OPT(I) \geq \max_{i\in[k]} X_i. \] 

Let $M:=\max_{i\in[k]} X_i$. We first show $\,\E{}{M}=O(\tfrac{\log k}{\log\log k})$. For any integer $t\ge 1$, by the union bound, \[ \Prob{M\ge t}\le k\cdot \Prob{X_1\ge t}. \] Since $X_1\sim\mathrm{Bin}(m,p)$ has mean $1$, a standard Chernoff bound yields $\Prob{X_1\ge t}\le (e/t)^t$ for all $t\ge 1$. Thus $\Prob{M\ge t}\le k(e/t)^t$ and therefore \[ \E{}{M}=\sum_{t\ge 1}\Prob{M\ge t} \le t_0 + \sum_{t>t_0} k\left(\frac{e}{t}\right)^t, \] where $t_0:=\left\lceil 3\frac{\log k}{\log\log k}\right\rceil$. For all sufficiently large $k$ and all $t\ge t_0$, the term $k(e/t)^t$ is at most $k^{-2}$, so the tail contributes $O(1)$ and hence \[ \E{}{M}=O\!\left(\frac{\log k}{\log\log k}\right). \] For the matching lower bound, set $t_1:=\left\lfloor \frac{\log k}{3\log\log k}\right\rfloor$. We lower bound $\Prob{X_1\ge t_1}$ by considering the event that a fixed subset of $t_1$ agents in $S_1$ are active: \[ \Prob{X_1\ge t_1}\ge \binom{m}{t_1} p^{t_1} (1-p)^{m-t_1} \ge \left(\frac{m}{t_1}\right)^{t_1} p^{t_1} \cdot (1-p)^m. \] Using $mp=1$ and $(1-p)^m\ge e^{-mp}=e^{-1}$, we get \[ \Prob{X_1\ge t_1}\ge e^{-1}\left(\frac{1}{t_1}\right)^{t_1}. \] Since $t_1=\Theta(\tfrac{\log k}{\log\log k})$, we have $\left(\tfrac{1}{t_1}\right)^{t_1}\ge k^{-1/3}$ for all sufficiently large $k$, and thus $\Prob{X_1\ge t_1}\ge k^{-1/3}/e$. By independence, \[ \Prob{M< t_1}=\left(1-\Prob{X_1\ge t_1}\right)^k \le \exp\!\left(-k\cdot \Prob{X_1\ge t_1}\right) \le \exp(-k^{2/3}/e), \] so $\Prob{M\ge t_1}\ge 1-\exp(-k^{2/3}/e)$ and therefore \[ \E{}{M}\ge t_1\cdot \Prob{M\ge t_1} =\Omega\!\left(\frac{\log k}{\log\log k}\right). \] Combining the two bounds yields $\E{}{M}=\Theta(\tfrac{\log k}{\log\log k})$. Finally, since $k=\lfloor n/m\rfloor=\Theta(n\frac{\log\log n}{\log n})$, we have $\log k=\Theta(\log n)$ and $\log\log k=\Theta(\log\log n)$, implying \[ \E{I\sim\mathcal{D}}{\OPT(I)} =\E{}{M} =\Theta\!\left(\frac{\log n}{\log\log n}\right), \] as claimed.

\end{document}